\documentclass[12pt, draftclsnofoot, onecolumn]{IEEEtran}

\usepackage{CJK}
\usepackage{color, soul, framed}
\usepackage{algorithm}
\usepackage{algorithmic}
\usepackage{amsmath}
\usepackage{amssymb}
\usepackage{psfrag}
\usepackage{graphicx}
\usepackage{epstopdf}
\usepackage{multirow}
\usepackage{stfloats}
\usepackage{cite}
\usepackage{color}
\usepackage[normalem]{ulem}
\usepackage{arydshln}
\usepackage{textcomp}
\usepackage{enumerate}
\usepackage{bbm}
\usepackage{xcolor}
\usepackage{xpatch}
\usepackage{array}
\usepackage{makecell}
\usepackage{multirow}
\usepackage{diagbox}
\usepackage{subfigure}

\newtheorem{mypro}{Proposition}
\newtheorem{remark}{\bf{Remark}}

\makeatletter

\begin{document}
\title{Joint Beamforming and Reconfigurable Intelligent Surface Design for Two-Way Relay Networks}
\author{Jun Wang, Ying-Chang Liang, \emph{Fellow, IEEE}, Jingon Joung, \emph{Senior Member, IEEE}, Xiaojun Yuan, \emph{Senior Member, IEEE}, and Xinguo Wang
\thanks{
Part of this work was presented in IEEE Globecom 2020\cite{wang2020joint}. This work has been submitted to the IEEE for possible publication. Copyright may be transferred without notice, after which this version may no longer be accessible.
\newline
\indent J. Wang is with the National Key Laboratory of Science and Technology on Communications, and the Center for Intelligent Networking and Communications (CINC), University of Electronic Science and Technology of China (UESTC), Chengdu 611731, China (e-mail: junwang@std.uestc.edu.cn). Y.-C. Liang and X.-J. Yuan are with the Center for Intelligent Networking and Communications (CINC), University of Electronic Science and Technology of China (UESTC), Chengdu 611731, China (e-mail: liangyc@ieee.org and xjyuan@uestc.edu.cn). J. Joung is with the School of Electrical and Electronics Engineering, Chung-Ang University, Seoul 06974, South Korea (e-mail: jgjoung@au.ac.kr). X.-G. Wang is with the School of Computer Science, Chengdu University of Information Technology, Chengdu 610025, China, and also with the Center for Intelligent Networking and Communications (CINC), and National Key Laboratory of Science and Technology on Communications, University of Electronic Science and Technology of China, Chengdu, China (e-mail: xinguowang911@163.com).}}

 \maketitle

\vspace{-1cm}
\begin{abstract}
In this paper, we consider a reconfigurable intelligent surface (RIS)-assisted two-way relay network, in which two users exchange information through the base station (BS) with the help of an RIS. By jointly designing the phase shifts at the RIS and beamforming matrix at the BS, our objective is to maximize the minimum signal-to-noise ratio (SNR) of the two users, under the transmit power constraint at the BS. We first consider the single-antenna BS case, and propose two algorithms to design the RIS phase shifts and the BS power amplification parameter, namely the SNR-upper-bound-maximization (SUM) method, and genetic-SNR-maximization (GSM) method. When there are multiple antennas at the BS, the optimization problem can be approximately addressed by successively solving two decoupled subproblems, one to optimize the RIS phase shifts, the other to optimize the BS beamforming matrix. The first subproblem can be solved by using SUM or GSM method, while the second subproblem can be solved by using optimized beamforming or maximum-ratio-beamforming method. The proposed algorithms have been verified through numerical results with computational complexity analysis.
\end{abstract}

\section{Introduction}
The \emph{sixth generation} (6G) mobile networks are expected to support peak data rate of terabits per second and millions of wireless connections per square kilometer~\cite{you2021towards,liang2020symbiotic,zhang20196g}. The exponential growth of the wireless traffic and communication device thus call for novel spectral- and energy-efficient technologies for future wireless communications~\cite{liang2020dynamic}. Recently, \emph{reconfigurable intelligent surface} (RIS), also known as an intelligent reflecting surface, has become a promising technique to help to fulfill these requirements~\cite{liang2019large,gong2019towards,8796365}. The rise of the RIS technique is closely related to the fast development of the meta-materials and the fabrication technology. Through intelligently adjusting the phase shifts of RIS, the wireless channels become programable and controllable. Thus RIS can be applied to various wireless communication systems to assist the performance improvement, and the study of RIS techniques has been attracting more and more attention from the industry and academia.

In \cite{huang2019reconfigurable}, RIS was used to achieve up to three times higher energy efficiency compared with conventional \emph{amplify-and-forward} (AF) relay. RIS was applied to wireless systems to achieve enhanced physical layer security in \cite{chen2019intelligent,guan2019intelligent}. The confidential data streams were transmitted to the legitimate receivers while keeping them secret from the eavesdroppers with the help of the RIS. RIS-aided multi-user downlink \emph{multiple-input single-output} (MISO) system was investigated in \cite{guo2020weighted}. The weighted sum rate was maximized by jointly designing the transmit beamforming and RIS phase shifts under the perfect and imperfect \emph{channel state information} (CSI) setup. The max-min fairness problem was considered in the RIS-aided multi-cell MISO systems \cite{xie2020max}. In \cite{xu2020resource}, the RIS-assisted multiuser full-duplex cognitive radio network was investigated. The secondary network employs a full-duplex BS to serve multiple half-duplex downlink and uplink secondary users simultaneously. Here, an RIS is deployed to improve the performance of the secondary network and mitigate the interference to the primary network. In \cite{han2019intelligent}, the power control problem was investigated for a physical-layer broadcasting scenario. The RIS-aided multi-group multicast MISO communication system was considered to maximize the sum rate of all the multicasting groups in \cite{zhou2019intelligent}. Also, the RIS-assisted symbiotic radio for an IoT communication system was proposed in \cite{zhang2020large}. RIS-assisted non-orthogonal multiple access system was studied in \cite{yang2019intelligent, ding2019simple}. In \cite{wang2019intelligent}, the joint active and passive precoding design for the RIS-assisted \emph{millimeter wave} (mmWave) communication was addressed by exploiting some important characteristics of mmWave channels for both single RIS and multi-RIS cases. The channel capacity optimization in indoor mmWave environments using RIS was studied in \cite{perovic2019channel}.

On the other hand, \emph{two-way relay network} (TWRN) is another promising strategy to improve the spectral efficiency in cooperative networks\cite{rankov2007spectral,liang2008optimal}. In TWRN, two phases are required to exchange information between two users through a \emph{two-way relay} (TWR). Here, the key challenge is how to design the beamforming matrix at the TWR. In \cite{zhang2009optimal} and its conference version \cite{liang2008optimal}, the authors derived the optimal structure of the beamforming matrix at the TWR and provided several sub-optimal schemes to maximize the sum-rate. The beamforming design for the \emph{multiple-input multiple-output} (MIMO) TWR communications to maximize the minimum end-to-end \emph{signal-to-noise ratio} (SNR) was investigated in \cite{wang2012maximin}. The problem was recast as a fractional programming problem and solved by using the Dinkelbach-type procedure combined with semi-definite programming. The beamforming design for multi-user TWR was studied in \cite{fang2013beamforming} through a max-min \emph{signal-to-interference-plus-noise ratio} (SINR) problem. The relay processing and power control problem for a multi-user two-way MIMO communication system was investigated in\cite{joung2009multiuser}.

The existing studies on RIS-assisted TWRN mainly exploited the RIS to replace the conventional relay to realize the information exchange between two users. Specifically, an RIS-assisted full-duplex MIMO TWRN system was investigated in \cite{zhang2020sum}, in which both users receive and transmit signals at the same time with the help of the RIS. In \cite{atapattu2020reconfigurable}, the authors analyzed the performance for reciprocal and non-reciprocal channels in the RIS-assisted TWRN system and derived the closed-form expressions for the outage probability and the spectral efficiency when the RIS is equipped with one reflective element. A more general case was studied in \cite{peng2020multiuser}, in which multiple full-duplex users exchange information with the full-duplex \emph{base station} (BS) with the assist of the RIS. As reported in reference \cite{pradhan2020reconfigurable}, however, the benefit of RIS is conditioned on the proper self-interference cancellation at the full-duplex TWR node, which brings high hardware complexity and low energy efficiency. In \cite{pradhan2020reconfigurable}, an RIS-enhanced two-way orthogonal-frequency-division multiplexing communication system with multiple pairs of users was investigated. By separating the available bandwidth into multiple orthogonal subbands and allocating them to the user pairs, the two-way device-to-device communication was accomplished.

In this paper, we investigate RIS-assisted TWRN, in which a multi-antenna BS serves two users to exchange information with the help of an RIS, where the BS and RIS can be considered as an active and passive TWRs, respectively. When any of the two links from the users to the BS is weak due to deep fading and shadowing, the RIS can enhance the weak link and provide fairness to both users. By intelligently reconfiguring the reflective elements on the RIS, the information exchange rate can be significantly improved. Here, our objective is to jointly design the beamforming and phase shift matrices at the BS and RIS, respectively, such that the minimum SNR of the two users is maximized under the transmit power constraint at the BS. When solving the optimization problem, the beamforming and phase shift matrices are coupled with each other, and the optimization problem is non-convex. To obtain a design insight, we first study the RIS phase shifts and power amplification parameter design for single-antenna BS case. We first propose the \emph{SNR-upper-bound-maximization} (SUM) algorithm which maximizes an upper bound of the original objective function, namely, the minimum of the combined channel gains seen by the two users. After that, an improved algorithm, called \emph{genetic-SNR-maximization} (GSM) algorithm, is proposed to solve the original problem approximately. When there are multiple antennas at the BS, the optimization problem can be approximately addressed by successively solving two decoupled subproblems, one to optimize the phase shift matrix at the RIS, one to optimize the beamforming matrix at the BS. The RIS phase shifts can be obtained by employing SUM or GSM method while the BS beamforming matrix can be obtained by \emph{optimized beamforming} (OB) or \emph{maximum-ratio-beamforming} (MRB) method. The main contributions of this study are summarized as follows.

\begin{itemize}
\item The RIS is applied as a passive TWR into TWRN to improve the information exchange rate.
\item The joint beamforming and RIS design problem for RIS-assisted TWRN is formulated to maximize the minimum SNR of the two users under the BS transmit power constraint.
\item To obtain the proper insight on the formulated optimization problem, the single-antenna BS case is first investigated. Since the problem is non-convex, we propose to maximize an upper bound of the original objective function, i.e., the minimum of the combined channel gains of the two users. The SUM and GSM algorithms are proposed to solve the design problem.
\item For the multiple-antenna BS case, the optimization problem is divided into two decoupled subproblems, one to optimize the RIS phase shifts, the other to optimize the BS beamforming matrix. The SUM and GSM methods are employed to obtain the RIS phase shifts while the OB and MRB methods are employed to obtain the BS beamforming matrix.

\item Numerical results verify that the proposed algorithms with RIS can improve the information exchange rate significantly in TWRN.
\end{itemize}

\begin{table}[htbp]\label{table:abbr}
\centering
\small{
\caption{List of abbreviations}

    \begin{tabular}{|l|l|}
    \cline{1-2}
    \textbf{Abbreviation} & \textbf{Description} \\
    \cline{1-2}
      AF &     Amplify-and-Forward \\
    \cline{1-2}
      AoA &  Angle of Arrival\\
    \cline{1-2}
      AoD & Angle of Departure\\
    \cline{1-2}
      BS &     Base Station \\
    \cline{1-2}
      CDF & Cumulative Distribution Function\\
    \cline{1-2}
      GSM & Genetic-SNR-Maximization\\
    \cline{1-2}
      GSM-MRB & Genetic-SNR-Maximization Maximum-Ratio-Beamforming\\
    \cline{1-2}
      GSM-OB & Genetic-SNR-Maximization Optimized-Beamforming\\
    \cline{1-2}
      MIMO  & Multiple-input Multiple-output   \\
    \cline{1-2}
      MISO &  Multiple-input Single-output \\
    \cline{1-2}
      mmWave & millimeter Wave \\
    \cline{1-2}
      MRR-MRT & Maximal-Ratio-Reception Maximal-Ratio-Transmission\\
    \cline{1-2}
      RIS & Reconfigurable Intelligent Surface\\
    \cline{1-2}
      SINR & Signal to Interference plus Noise Ratio\\
    \cline{1-2}
      SNR & Signal to Noise Ratio\\
    \cline{1-2}
      SUM & SNR-Upper-bound-Maximization\\
    \cline{1-2}
      SUM-MRB & SNR-Upper-bound-Maximization Maximum-Ratio-Beamforming\\
    \cline{1-2}
      SUM-OB & SNR-Upper-bound-Maximization Optimized-Beamforming\\
    \cline{1-2}
      TWR & Two-Way Relay\\
    \cline{1-2}
      TWRN & Two-Way Relay Network\\
    \cline{1-2}
      6G & Sixth Generation \\
    \cline{1-2}
    \end{tabular}
  
  }
\end{table}

The rest of this paper is organized as follows. Section~\ref{systemmodel} presents the system model for the RIS-assisted two-way relay network and provides a joint design optimization problem. In Section~\ref{secspecialcase}, the single-antenna BS case is presented to obtain insights into the original problem. In Section~\ref{secproposed1}, the algorithms to solve a problem for the multiple-antenna BS case are proposed. Section~\ref{secsimulation} provides simulation results to validate the effectiveness of the proposed algorithms. Section~\ref{secconclusion} concludes this study.

\emph{Notations}: For complex vector $\boldsymbol{v}$, $\boldsymbol{v}^*$, $\boldsymbol{v}^T$, $\boldsymbol{v}^H$, and diag($\boldsymbol{v}$) denote the conjugate, the transpose, the conjugate transpose, and the diagonal matrix with its diagonal elements given by $\boldsymbol{v}$. Scalar $\boldsymbol{v}_i$ denotes the \emph{i}th element of vector $\boldsymbol{v}$. $[\boldsymbol{v}]_{(1:N)}$ denotes the first $N$ elements of vector $\boldsymbol{v}$. $\boldsymbol{a}\otimes \boldsymbol{b}$ denotes the Kronecker product of vector $\boldsymbol{a}$ and $\boldsymbol{b}$. $\text{vec}(\boldsymbol{A})$ denotes the vectorization operation for matrix $\boldsymbol{A}$. $\boldsymbol{A}^{\star}$ denotes the optimal value of variable $\boldsymbol{A}$. $\text{tr}(\boldsymbol{A})$ and $\text{rank}(\boldsymbol{A})$ denote the trace and rank of the matrix $\boldsymbol{A}$, respectively. $\boldsymbol{A}[m,n]$ and $\boldsymbol{A}[m,:]$ denote the $(m,n)$th element and the $m$th row vector of matrix $\boldsymbol{A}$, respectively. $\boldsymbol{A}\succeq 0$ denotes that matrix $\boldsymbol{A}$ is a semi-definite matrix. The distribution of a \emph{circularly symmetric complex Gaussian} (CSCG) random variable with mean $\mu$ and variance $\sigma^2$ is denoted by ${\cal{CN}}(\mu,\sigma^2)$. Finally, the list of abbreviations appeared in this paper is given in Table I.

\section{System Model}\label{systemmodel}

\begin{figure}[!t]
\centering
\subfigure[] {
 \label{fig:a}
\includegraphics[width=0.4\textwidth]{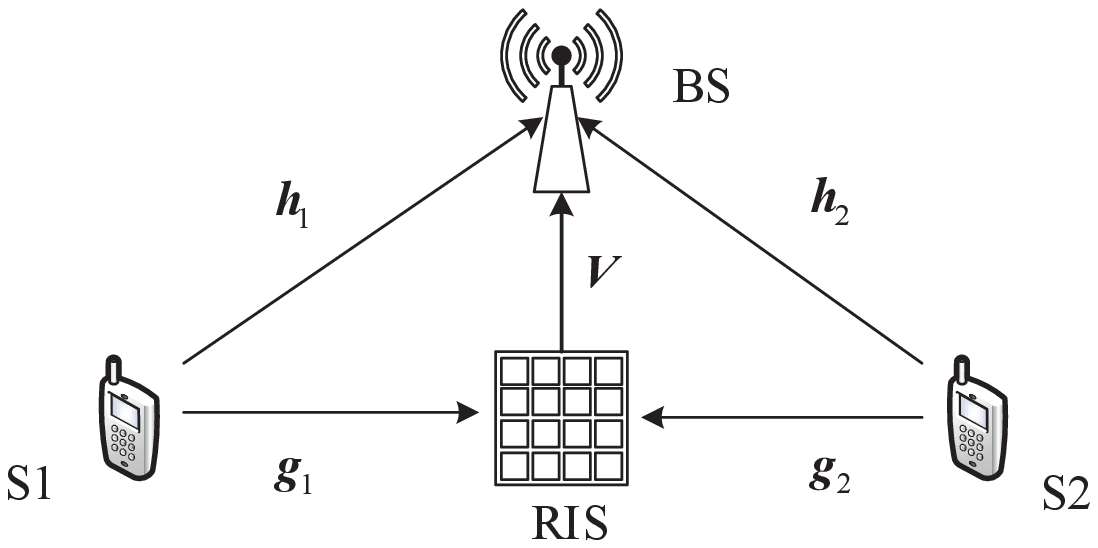}
}
\subfigure[]{
\label{fig:b}
\includegraphics[width=0.4\textwidth]{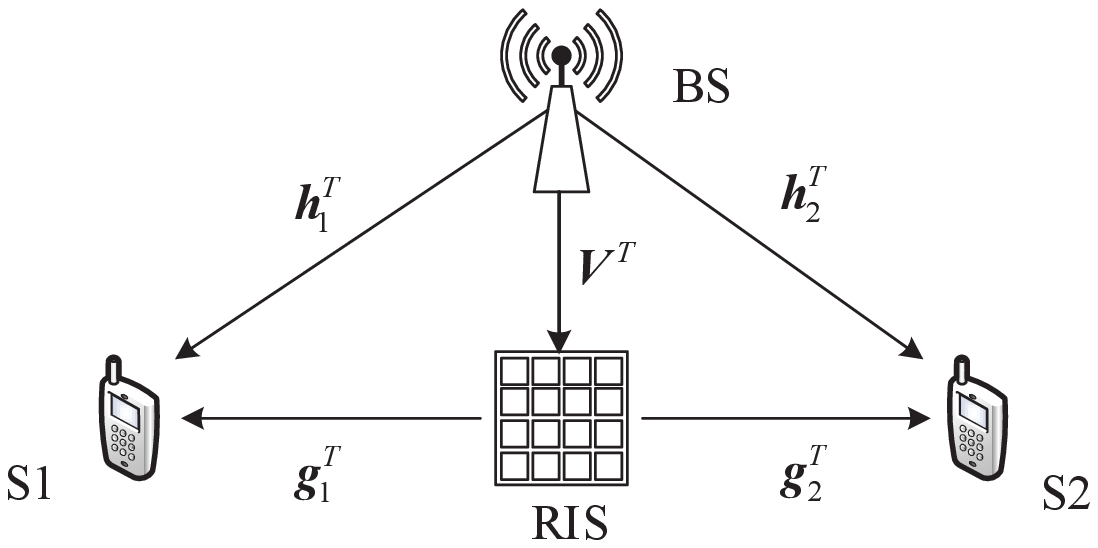}
}
\caption{ The system model for the RIS-assisted two-way relay network. (a) First phase: two users send information to the BS with the help of an RIS; (b) Second phase: The BS broadcasts the processed signal to the users with the help of an RIS.}
\label{fig:Fig0}
\end{figure}

In this paper, we consider an RIS-assisted TWRN. As illustrated in Fig. \ref{fig:Fig0}, the system consists of a BS with $M$ antennas, two single-antenna users, i.e., the source nodes denoted by S1 and S2, and an RIS with $N$ reflective elements, each of which can introduce a phase shift to the incident signal. Two users exchange information with the help of the BS and RIS, where the BS and RIS operate as active and passive TWRs, respectively. The RIS is deployed to enhance the TWR channels between the users and the BS so that the information exchange rate can be improved. In the following, we present the channel and the signal models of the RIS-assisted TWRN considered in this study.
\subsection{Channel Model}
We assume the flat fading channels, which means that the channels remain unchanged during one transmission block. The channels from S1 and S2 to the BS are respectively denoted by $\boldsymbol{h}_1\in \mathbb{C}^{M\times 1}$ and $\boldsymbol{h}_2\in \mathbb{C}^{M\times 1}$, which are modeled as the Rayleigh channels. Here, it is assumed that the \emph{line-of-sight} (LoS) path is blocked. The elements of the channels are independent and follow a distribution of $\sqrt{\eta(d_{\text{B},k})}{\cal{CN}}(0,1),k=1,2$, where $\eta(d_{\text{B},k})$ denotes the large-scale path loss component of the channels depending on the distance $d_{\text{B},k}$ between user $k$ and the BS. The channels from S1 to RIS, from S2 to RIS, and from RIS to BS are denoted by $\boldsymbol{g}_1 \in \mathbb{C}^{N\times 1}$, $\boldsymbol{g}_2 \in \mathbb{C}^{N\times 1}$, and $\boldsymbol{V} \in \mathbb{C}^{M\times N}$, respectively. Since the LoS path exists in these channels, the channels are modeled as the Rician channels, without loss of generality, i.e.,
\begin{align}
\boldsymbol{V}=\sqrt{\eta(d_{\text{B,R}})}\left(\sqrt{\frac{K_v}{1+K_v}}\boldsymbol{V}^{\text{LoS}}+\sqrt{\frac{1}{1+K_v}}\boldsymbol{V}^{\text{NLoS}}\right),
\end{align}
and
\begin{align}
\boldsymbol{g}_k=\sqrt{\eta(d_{k,\text{R}})}\left(\sqrt{\frac{K_k}{1+K_k}}\boldsymbol{g}_k^{\text{LoS}}+\sqrt{\frac{1}{1+K_k}}\boldsymbol{g}_k^{\text{NLoS}}\right),k=1,2,
\end{align}
 where $d_{\text{B,R}}$ is the distance between RIS and BS, and $d_{k,\text{R}}$ is the distance between user \emph{k} and RIS. The small-scale component consists of the LoS and non-LoS (NLoS) components. $K_v$ and $K_k$ are the Rician factors of $\boldsymbol{V}$ and $\boldsymbol{g}_k$, respectively. $\boldsymbol{V}^{\text{LoS(NLoS)}}$ and $\boldsymbol{g}_k^{\text{LoS(NLoS)}}$ denote the LoS (NLoS) components of $\boldsymbol{V}$ and $\boldsymbol{g}_k$, respectively. The NLoS components follow the standard complex Gaussian distribution with zero mean and unit variance. The LoS components can be expressed by the responses of the RIS.

 Since the RIS is a uniform rectangular array, the steering vector $\boldsymbol{a}_R(\theta,\psi)$ at the RIS is modeled as follows:
\begin{align}
\boldsymbol{a}_R\left(\theta,\psi\right)=\boldsymbol{a}_v\left(\theta,\psi\right)\otimes \boldsymbol{a}_h\left(\theta,\psi\right)\in\mathbb{C}^{1\times N},\label{steering_RIS}
\end{align}
where $\theta$ and $\psi$ denote the center azimuth and elevation angles, respectively, of the arriving or departing signals at the RIS; $\boldsymbol{a}_h\left(\theta,\psi\right)\in\mathbb{C}^{1 \times N_h}$ and $\boldsymbol{a}_v\left(\theta,\psi\right)\in\mathbb{C}^{1 \times N_v}$ are the steering vectors in the horizontal and vertical directions, respectively\cite{liu2019super}; $N_h$ and $N_v$ are the numbers of elements along the horizontal and vertical axes, respectively. Here, the elements of $\boldsymbol{a}_h$ and $\boldsymbol{a}_v$ can be modeled, respectively, as
\begin{align}
\left[\boldsymbol{a}_h\left(\theta,\psi\right)\right]_n=e^{-j\frac{2 \pi d}{\lambda}(n-1)\text{cos}\left(\psi\right)\text{sin}\left(\theta\right)},\forall{n} \in \left\{1,2,\ldots,N_h\right\},
\end{align}
\begin{align}
\left[\boldsymbol{a}_v\left(\theta,\psi\right)\right]_n=e^{j\frac{2 \pi d}{\lambda}(n-1)\text{cos}\left(\psi\right)\text{cos}\left(\theta\right)},\forall{n} \in \left\{1,2,\ldots,N_v\right\},
\end{align}
where $d$ and $\lambda$ are the antenna element separation and carrier wavelength, respectively. On the other hand, since BS employs a horizontal linear array, i.e., $\psi=0$, the steering vector at the BS is modeled as
\begin{align}
\boldsymbol{a}_B\left(\theta_{{\rm AoA},\text{B}}\right)=[1,e^{-j\frac{2\pi d}{\lambda}\text{sin}\left(\theta_{{\rm AoA},\text{B}}\right)},...,e^{-j\frac{2\pi d}{\lambda}(M-1)\text{sin}\left(\theta_{{\rm AoA},\text{B}}\right)}]\in\mathbb{C}^{1\times M},\label{steering_BS}
\end{align}
where $\theta_{{\rm AoA},\text{B}}$ denotes the {\it angle of arrival} (AoA) at the BS. $d$ is set as $d=\frac{1}{2}\lambda$ for simplicity.

Using the steering vectors in (\ref{steering_RIS}) and (\ref{steering_BS}), the LoS channels are modeled as follows:
\begin{align}
\boldsymbol{V}^{\text{LoS}}=\boldsymbol{a}_B^H\left(\theta_{{\rm AoA},\text{B}}\right)\boldsymbol{a}_R\left(\theta_{{\rm AoD},\text{R}},\psi_{{\rm AoD},\text{R}}\right),
\end{align}
where $\theta_{{\rm AoD},\text{R}}$ and $\psi_{{\rm AoD},\text{R}}$ denote the center azimuth \emph{angle of departure} (AoD) and elevation AoD, respectively, at the RIS and
\begin{align}
\boldsymbol{g}_k^{\text{LoS}}=\boldsymbol{a}_R^H(\theta_{{\rm AoA},k},\psi_{{\rm AoA},k}),k=1,2,
\end{align}
where $\theta_{{\rm AoA},k}$ and $\psi_{{\rm AoA},k}$ denote the center azimuth and elevation AoAs, respectively, at the RIS for user $k$.
\vspace{-0.6cm}
\subsection{Signal Model}
Under a time-division-duplex mode, two phases are needed for the two users to exchange information. In the first phase, two users send their information to the BS, simultaneously, with the help of an RIS as shown in Fig. 2(a). In the second phase, the BS performs beamforming for the signal received in the first phase and broadcasts the processed signal to the users, as shown in Fig. 2(b). Two users then decode the received signals, and the information exchange is completed.

In the first phase, the BS received signal is given by
\begin{align}
\boldsymbol{r}=\sqrt{P_S}\left(\boldsymbol{h}_1+\boldsymbol{V\Phi}_1\boldsymbol{g}_1\right)x_1+\sqrt{P_S}\left(\boldsymbol{h}_2+\boldsymbol{V\Phi}_1\boldsymbol{g}_2\right)x_2+\boldsymbol{u},
\end{align}
where $\boldsymbol{\Phi}_1$ denotes a diagonal phase shift matrix at RIS at the first phase; $P_S$ is the transmit power of users; and $\boldsymbol{u}\in\mathbb{C}^{M\times 1} \sim{\cal{CN}}(0,\sigma^2\boldsymbol{I})$ is the additive complex noise vector; and $x_1$ and $x_2$ denote the information signal of S1 and S2, respectively. Here, it is assumed that $x_1$ and $x_2$ conform to the same distribution ${\cal{CN}}(0,1)$. After receiving the signal $\boldsymbol{r}$, the BS performs beamforming to generate the retransmit signal as follows:
\begin{align}
\boldsymbol{s}=\boldsymbol{A}\boldsymbol{r},
\end{align}
where $\boldsymbol{A}\in \mathbb{C}^{M\times M}$ denotes the beamforming matrix$\footnote{The BS acts as an AF relay in the TWRN since the AF relay requires much less delay and computing power compared to a decode-and-forward relay\cite{levin2012amplify}.}$.

In the second phase, assuming that the channel reciprocity holds during the first and second phases, the second-phase channel can be modeled as the transpose of the first-phase channel. Denoting the phase shift matrix at the second phase by $\boldsymbol{\Phi}_2$, the received signal at S1 is then written as follows:
\begin{align}
y_1=&(\boldsymbol{h}_1+\boldsymbol{V\Phi}_2\boldsymbol{g}_1)^T\boldsymbol{s}+z_1\nonumber\\
=&\sqrt{P_S}(\boldsymbol{h}_1+\boldsymbol{V\Phi}_2\boldsymbol{g}_1)^T\boldsymbol{A}(\boldsymbol{h}_1+\boldsymbol{V\Phi}_1\boldsymbol{g}_1)x_1+\sqrt{P_S}(\boldsymbol{h}_1+\boldsymbol{V\Phi}_2\boldsymbol{g}_1)^T\boldsymbol{A}(\boldsymbol{h}_2+\boldsymbol{V\Phi}_1\boldsymbol{g}_2)x_2\nonumber\\
&+(\boldsymbol{h}_1+\boldsymbol{V\Phi}_2\boldsymbol{g}_1)^T\boldsymbol{A}\boldsymbol{u}+z_1.
\end{align}

Since S1 knows its own signal and CSI $(\boldsymbol{h}_1+\boldsymbol{V\Phi}_2\boldsymbol{g}_1)^T\boldsymbol{A}(\boldsymbol{h}_1+\boldsymbol{V\Phi}_1\boldsymbol{g}_1)$ following the signaling and channel estimation in \cite{gao2009optimal}, $x_1$ can be canceled from the received signal, i.e., self-interference cancellation, yielding
\begin{align}
\widetilde{y}_1=\sqrt{P_S}(\boldsymbol{h}_1+\boldsymbol{V\Phi}_2\boldsymbol{g}_1)^T\boldsymbol{A}(\boldsymbol{h}_2+\boldsymbol{V\Phi}_1\boldsymbol{g}_2)x_2+(\boldsymbol{h}_1+\boldsymbol{V\Phi}_2\boldsymbol{g}_1)^T\boldsymbol{A}\boldsymbol{u}+z_1.
\end{align}
Similarly, the signal received by S2 after self-interference cancellation is given by
\begin{align}
\widetilde{y}_2=\sqrt{P_S}(\boldsymbol{h}_2+\boldsymbol{V\Phi}_2\boldsymbol{g}_2)^T\boldsymbol{A}(\boldsymbol{h}_1+\boldsymbol{V\Phi}_1\boldsymbol{g}_1)x_1+(\boldsymbol{h}_2+\boldsymbol{V\Phi}_2\boldsymbol{g}_2)^T\boldsymbol{A}\boldsymbol{u}+z_2.
\end{align}
Here, $\boldsymbol{u}$, $z_1$, and $z_2$ are independent, where $z_1$ and $z_2$ conform to the distribution ${\cal{CN}}(0,\sigma^2)$.

The SNR at S1 can then be derived as follows:
\begin{align}
\gamma_1=\frac{P_S\left|(\boldsymbol{h}_1+\boldsymbol{V\Phi}_2\boldsymbol{g}_1)^T\boldsymbol{A}(\boldsymbol{h}_2+\boldsymbol{V\Phi}_1\boldsymbol{g}_2)\right|^2}{\left\|(\boldsymbol{h}_1+\boldsymbol{V\Phi}_2\boldsymbol{g}_1)^T\boldsymbol{A}\right\|^2\sigma^2+\sigma^2}=\frac{\beta\left|(\boldsymbol{h}_1+\boldsymbol{V\Phi}_2\boldsymbol{g}_1)^T\boldsymbol{A}(\boldsymbol{h}_2+\boldsymbol{V\Phi}_1\boldsymbol{g}_2)\right|^2}{\left\|(\boldsymbol{h}_1+\boldsymbol{V\Phi}_2\boldsymbol{g}_1)^T\boldsymbol{A}\right\|^2+1}\label{eqp4},
\end{align}
where $\beta=P_S/\sigma^2$. Similarly, the SNR at S2 is
\begin{align}
\gamma_2=\frac{\beta\left|(\boldsymbol{h}_2+\boldsymbol{V\Phi}_2\boldsymbol{g}_2)^T\boldsymbol{A}(\boldsymbol{h}_1+\boldsymbol{V\Phi}_1\boldsymbol{g}_1)\right|^2}{\left\|(\boldsymbol{h}_2+\boldsymbol{V\Phi}_2\boldsymbol{g}_2)^T\boldsymbol{A}\right\|^2+1}\label{eqp5}.
\end{align}
We have the following proposition.
\begin{mypro}\label{Proposition:Prop1}
During consecutive two-phase transmissions, we design the phase shifts, such that they do not vary over the two phases, i.e. $\boldsymbol{\Phi}_1=\boldsymbol{\Phi}_2=\boldsymbol{\Phi}=\text{diag}\left(e^{j\theta_1},e^{j\theta_2},...,e^{j\theta_N}\right)$ where $\theta_i$ denotes the phase shift introduced by the $i$th RIS element.
\end{mypro}
\begin{proof}
  See proof in Appendix \ref{app:prop1}.
\end{proof}
The transmit power at BS is given by
\begin{align}
P(\boldsymbol{A},\boldsymbol{\Phi})=&\text{tr}\left[\boldsymbol{s}\boldsymbol{s}^H\right]=P_S\left\|\boldsymbol{A}(\boldsymbol{h}_1+\boldsymbol{V\Phi}\boldsymbol{g}_1)\right\|^2+P_S\left\|\boldsymbol{A}(\boldsymbol{h}_2+\boldsymbol{V\Phi}\boldsymbol{g}_2)\right\|^2+\sigma^2\text{tr}(\boldsymbol{A}\boldsymbol{A}^H).
\end{align}
\begin{remark}
For the more general case in which a large number of users must be served, two users are paired each time to realize an exchange of information. By allocating orthogonal resources, such as time and frequency, to different two-user pairs, i.e., TDMA and OFDMA, the proposed algorithms for a single pair of two users in the two-way communications can be directly extended to the multiple two-user pairs. Without the assistance of the RIS, the scenario is conventional multipair two-way relay network and there are already several researches working on this, such as \cite{wang2016multipair}.
\end{remark}
\subsection{Discussion on Channel State Information Acquisition}
To optimize the beamforming matrix $\boldsymbol{A}$ and the phase shift matrix $\boldsymbol{\Phi}$, the RIS-assisted TWRN system needs to acquire the CSI. In the literature, the CSI of the RIS channels can be obtained through, e.g., the Brute-Force method\cite{nadeem2019intelligent}, the compressive-sensing method \cite{he2019cascaded}, and the semi-passive RIS method\cite{taha2019enabling}. From the channel reciprocity during the uplink and downlink, the second-phase channel can be obtained by taking the transpose of the first-phase channel. Here, the channel estimation error can be tackled by using a robust design method, e.g., \cite{wang2020robust}.

\subsection{Problem Formulation}\label{secformulation}
To maximize the information exchange rate while considering fairness between two users, the minimum SNR of the two users is maximized by jointly optimizing beamforming matrix $\boldsymbol{A}$ and phase shift matrix $\boldsymbol{\Phi}$ at BS and RIS, respectively, subject to the transmit power constraint at the BS. Denoting the transmit power budget at the BS by $P_B$, the optimization problem is formulated as follows:
\begin{subequations}
\label{eq:Po}
\begin{align}
{\bf (Po) }: \underset{\boldsymbol{\Phi},\boldsymbol{A}}{\max}  \quad  &\min\left\{\gamma_1,\gamma_2\right\}\nonumber \\
\text{s.t.}\quad
&|\phi_n|=1,\forall n,\label{eqp2} \\
&\beta\left\|\boldsymbol{A}(\boldsymbol{h}_1+\boldsymbol{V\Phi}\boldsymbol{g}_1)\right\|^2+\beta\left\|\boldsymbol{A}(\boldsymbol{h}_2+\boldsymbol{V\Phi}\boldsymbol{g}_2)\right\|^2+\text{tr}(\boldsymbol{A}\boldsymbol{A}^H)\leq \frac{P_B}{\sigma^2}.\label{eqp3}
\end{align}
\end{subequations}
This problem is challenging to be directly solved because BS beamforming matrix $\boldsymbol{A}$ and phase shift matrix $\boldsymbol{\Phi}$ are coupled to each other. The quadratic term of $\boldsymbol{\Phi}$ in the numerator of the SNR in (\ref{eqp4}) and (\ref{eqp5}) makes the problem even more intractable. To solve the optimization problem ${\bf (Po)}$, henceforth, we propose several efficient algorithms, which are as shown in Table \ref{table_Proposed algorithms}.
\begin{table*}
  \centering
  \fontsize{9}{12}\selectfont
  \small{
  \caption{Proposed algorithms throughout the paper}
  \label{table_Proposed algorithms}
  {{
    \begin{tabular}{|c|p{3cm}<{\centering}|p{3cm}<{\centering}|p{5cm}<{\centering}|}
    \cline{1-4}
    {\diagbox{Scenarios}{Algorithms}} & {Optimizing $\boldsymbol{\Phi}$} &{Optimizing $\boldsymbol{A}$}&{Section} \\
    \cline{1-4}
    \multirow{2}{*}{Single-antenna BS case}
                       & SUM     &\multirow{2}{*}{Formula (21)} &Section III-A: SUM\cr\cline{2-2}\cline{4-4}
                       & GSM    & &Section III-B: GSM\cr
                       \cline{1-4}

    \multirow{4}{*}{Multiple-antenna BS case}
             &\multirow{2}{*}{SUM}         &OB  &Section IV-A: SUM-OB \cr\cline{3-4}
                                          & &MRB &Section IV-B: SUM-MRB\cr\cline{2-4}
    &\multirow{2}{*}{GSM}    &OB  &Section IV-C: GSM-OB \cr\cline{3-4}
                          & &MRB &Section IV-C: GSM-MRB\cr\cline{1-4}
    \end{tabular}
    }
    }
    }
\end{table*}
\section{RIS for Single-Antenna BS}\label{secspecialcase}
To obtain some insights on how to solve the problem ${\bf (Po)}$ in (\ref{eq:Po}) and design the system, we first study the single-antenna BS case. In this case, the channel matrix $\boldsymbol{V} \in \mathbb{C}^{M\times N}$ from RIS to BS degenerates to a vector, denoted by $\boldsymbol{v} \in \mathbb{C}^{1\times N}$. Moreover, beamforming matrix $\boldsymbol{A}$ reduces to a power amplification parameter $\tau$, which means that the BS amplifies the received signal without beamforming. The original problem is thus simplified to a problem of how to obtain the RIS phase shift matrix, $\boldsymbol{\Phi}$, and BS power amplification parameter $\tau$. For the single-antenna BS case, the BS power constraint in (\ref{eqp3}) is written as
\begin{align}
\tau(P_S|h_1+\boldsymbol{v\Phi}\boldsymbol{g}_1|^2+P_S|h_2+\boldsymbol{v\Phi}\boldsymbol{g}_2|^2+\sigma^2)\leq P_B.\label{single_power_constraint}
\end{align}
The SNRs at S1 and S2 in (\ref{eqp4}) and (\ref{eqp5}) are then derived respectively as follows:
\begin{align}
\gamma_{S1}=\frac{\beta \tau|(h_1+\boldsymbol{v\Phi}\boldsymbol{g}_1)^T(h_2+\boldsymbol{v\Phi}\boldsymbol{g}_2)|^2}{\tau|h_1+\boldsymbol{v\Phi}\boldsymbol{g}_1|^2+1},\label{gamma_1_single}\\
\gamma_{S2}=\frac{\beta \tau|(h_2+\boldsymbol{v\Phi}\boldsymbol{g}_2)^T(h_1+\boldsymbol{v\Phi}\boldsymbol{g}_1)|^2}{\tau|h_2+\boldsymbol{v\Phi}\boldsymbol{g}_2|^2+1}.\label{gamma_2_single}
\end{align}

Since the SNRs of both users increase as $\tau$ increases, for a given phase shift matrix $\boldsymbol{\Phi}$, the optimal power amplification parameter is obtained from the equality in (\ref{single_power_constraint}) as
\begin{align}
\tau=\frac{P_B}{P_S|h_1+\boldsymbol{v\Phi}\boldsymbol{g}_1|^2+P_S|h_2+\boldsymbol{v\Phi}\boldsymbol{g}_2|^2+\sigma^2}.\label{obtain_tau}
\end{align}

Using (\ref{single_power_constraint})--(\ref{gamma_2_single}), problem ${\bf (Po)}$ for the single-antenna BS case is then rewritten as follows:
\begin{subequations}
\label{eq:Ps-1}
\begin{align}
{\bf (Ps-1) }: \underset{\boldsymbol{\Phi}}{\max}  \quad  &\min\left\{\gamma_{S1},\gamma_{S2}\right\}\nonumber \\
\text{s.t.}\quad
&|\phi_n|=1,\forall n. \tag{22}
\end{align}
\end{subequations}

Problem ${\bf (Ps-1)}$ is still intractable because the quadratic term of the phase shift matrix still exists in the SNR. Thus, we consider the upper bounds of the SNRs as follows:
\begin{align}
\gamma_{S1}&=\frac{\beta \tau|h_1+\boldsymbol{v\Phi}\boldsymbol{g}_1|^2|h_2+\boldsymbol{v\Phi}\boldsymbol{g}_2|^2}{\tau|h_1+\boldsymbol{v\Phi}\boldsymbol{g}_1|^2+1}=\frac{\beta |h_1+\boldsymbol{v\Phi}\boldsymbol{g}_1|^2|h_2+\boldsymbol{v\Phi}\boldsymbol{g}_2|^2}{|h_1+\boldsymbol{v\Phi}\boldsymbol{g}_1|^2+\frac{1}{\tau}}\nonumber\\
&=\frac{\beta |h_1+\boldsymbol{v\Phi}\boldsymbol{g}_1|^2|h_2+\boldsymbol{v\Phi}\boldsymbol{g}_2|^2}{|h_1+\boldsymbol{v\Phi}\boldsymbol{g}_1|^2+ \frac{P_S|h_1+\boldsymbol{v\Phi}\boldsymbol{g}_1|^2+P_S|h_2+\boldsymbol{v\Phi}\boldsymbol{g}_2|^2+\sigma^2}{P_B}}=\frac{\beta|h_2+\boldsymbol{v\Phi}\boldsymbol{g}_2|^2}{1+\frac{P_S}{P_B}+\frac{P_S}{P_B}\frac{|h_2+\boldsymbol{v\Phi}\boldsymbol{g}_2|^2}{|h_1+\boldsymbol{v\Phi}\boldsymbol{g}_1|^2}+\frac{\sigma^2}{P_B|h_1+\boldsymbol{v\Phi}\boldsymbol{g}_1|^2}}\nonumber\\
&\leq \beta |  h_2 +{\boldsymbol v}{\boldsymbol \Phi}{\boldsymbol g}_2 |^2 \triangleq \bar{\gamma}_{S1}.\label{gamma_s1}
\end{align}
Similarly, the upper bound of $\gamma_{S2}$ is derived as $\bar{\gamma}_{S2} \triangleq \beta|h_1+\boldsymbol{v\Phi}\boldsymbol{g}_1|^2$. Because the upper bounds of SNRs are tight when the transmit power of the BS is much greater than that of users, i.e., $P_B\gg P_S$, and this is a typical case of TWRN, the upper bound of SNRs can be used to design the TWRN. Meanwhile, the optimality loss does not highly depend on the tightness of the upper bound. Thus, we devise two algorithms to solve ${\bf (Ps-1)}$ by maximizing the upper bound of SNRs.
\subsection{SNR-Upper-bound-Maximization (SUM) Algorithm}
The proposed SUM algorithm solves the following problem:
\begin{subequations}
\label{eq:Ps-2}
\begin{align}
{\bf (Ps-2) }: \underset{\boldsymbol{\Phi}}{\max}  \quad  &\min\left\{\bar{\gamma}_{S1}, \bar{\gamma}_{S2}\right\}\nonumber \\
\text{s.t.}\quad
&|\phi_n|=1,\forall n.\tag{24}
\end{align}
\end{subequations}
Denoting a phase shift vector as $\boldsymbol{\phi}=\left[e^{j\theta_1},e^{j\theta_2},...,e^{j\theta_N}\right]^T\in\mathbb{C}^{N\times 1}$, the combined channel $h_1+\boldsymbol{v\Phi}\boldsymbol{g}_1$ can be written as
\begin{align}
h_1+\boldsymbol{v\Phi}\boldsymbol{g}_1=h_1+\boldsymbol{v}\text{diag}(\boldsymbol{g}_1)\boldsymbol{\phi}=[\boldsymbol{v}\text{diag}(\boldsymbol{g}_1),h_1][\boldsymbol{\phi}^T,1]^T\triangleq \bar{\boldsymbol{g}}_1^H\bar{\boldsymbol{\phi}},
\end{align}
where $[\boldsymbol{v}\text{diag}(\boldsymbol{g}_1),h_1]\triangleq \bar{\boldsymbol{g}}_1^H \in\mathbb{C}^{ 1 \times (N+1)}$ and $[\boldsymbol{\phi}^T,1]^T\triangleq \bar{\boldsymbol{\phi}}\in\mathbb{C}^{ (N+1) \times 1}$. Similarly,
\begin{align}
h_2+\boldsymbol{v\Phi}\boldsymbol{g}_2=h_2+\boldsymbol{v}\text{diag}(\boldsymbol{g}_2)\boldsymbol{\phi}=[\boldsymbol{v}\text{diag}(\boldsymbol{g}_2),h_2][\boldsymbol{\phi}^T,1]^T \triangleq\bar{\boldsymbol{g}}_2^H\bar{\boldsymbol{\phi}}.
\end{align}
The problem ${\bf (Ps-2)}$ is then equivalently transformed to
\begin{subequations}
\label{eq:Ps-3}
\begin{align}
{\bf (Ps-3) }: \underset{\bar{\boldsymbol{\phi}}}{\max}  \quad  &\min\left\{|\bar{\boldsymbol{g}}_1^H\bar{\boldsymbol{\phi}}|^2,|\bar{\boldsymbol{g}}_2^H\bar{\boldsymbol{\phi}}|^2\right\}\nonumber \\
\text{s.t.}\quad
&|\phi_n|=1,\forall n.\tag{27}
\end{align}
\end{subequations}
By introducing an additional variable, denoted by $t$, the problem can be recast as follows:
\begin{subequations}
\label{eq:Ps-4}
\begin{align}
{\bf (Ps-4) }: \underset{\bar{\boldsymbol{\phi}},t}{\max}  \quad  &t\nonumber \\
\text{s.t.}\quad
&|\phi_n|=1,\forall n, \label{phi_constraint}\\
&\bar{\boldsymbol{g}}_1^H\bar{\boldsymbol{\phi}}\bar{\boldsymbol{\phi}}^H\bar{\boldsymbol{g}}_1 \geq t,\label{ps-4-second}\\
&\bar{\boldsymbol{g}}_2^H\bar{\boldsymbol{\phi}}\bar{\boldsymbol{\phi}}^H\bar{\boldsymbol{g}}_2 \geq t.\label{ps-4-third}
\end{align}
\end{subequations}
Defining $\boldsymbol{\Psi}\triangleq \bar{\boldsymbol{\phi}}\bar{\boldsymbol{\phi}}^H\in \mathbb{C}^{(N+1)\times (N+1)}$ where $\boldsymbol{\Psi} \succeq 0$ and $\text{rank}(\boldsymbol{\Psi})=1$, the semi-definite constraints in (\ref{ps-4-second}) and (\ref{ps-4-third}) become convex. Using $\boldsymbol{\Psi}$ and relaxing the rank-one constraint for $\boldsymbol{\Psi}$, we can solve the following SDP problem:
\begin{subequations}
\label{eq:Ps-5}
\begin{align}
{\bf (Ps-5) }: \underset{\boldsymbol{\Psi},t}{\max}  \quad  &t\nonumber \\
\text{s.t.}\quad
&\boldsymbol{\Psi}[n,n]=1,\forall n,\label{new_phi_constraint} \\
&\boldsymbol{\Psi} \succeq 0,\\
&\text{tr}[\boldsymbol{\Psi}\bar{\boldsymbol{g}}_1\bar{\boldsymbol{g}}_1^H] \geq t,\\
&\text{tr}[\boldsymbol{\Psi}\bar{\boldsymbol{g}}_2\bar{\boldsymbol{g}}_2^H] \geq t.
\end{align}
\end{subequations}
This SDP problem can be solved efficiently via CVX\cite{grant2014cvx}.
The optimal solution of $\boldsymbol{\Psi}$, however, is not generally a rank-one matrix. Therefore, after obtaining the optimal $\boldsymbol{\Psi}$ from ${\bf (Ps-5)}$, we need to find a rank-one solution by using the Gaussian randomization procedure as summarized in Algorithm 1. Once the optimal solution $\bar{\boldsymbol{\phi}}^{\star}$ is obtained, we obtain $\boldsymbol{\phi}^{\star}$ as
\begin{align}
\boldsymbol{\phi}^{\star}=e^{j \text{arg}\left(\left[\frac{\bar{\boldsymbol{\phi}}^{\star}}{\bar{\phi}_{N+1}}\right]_{\left(1:N\right)}\right)},\label{gainphi*}
\end{align}
where $\text{arg}(x)$ denotes the operation of taking the angle of the complex value $x$. The overall SUM algorithm is summarized in Algorithm 2.

\begin{algorithm}[t!]
\caption{ Gaussian randomization procedure for obtaining the rank-one solution}\label{AlgorithmSDR}
\begin{algorithmic}[1]
\STATE Input: {The solution of $\bf (Ps-5)$: $\boldsymbol{\Psi}^{\star}$}
\STATE Perform singular value decomposition for $\boldsymbol{\Psi}^{\star}$ as $\boldsymbol{\Psi}^{\star}=\boldsymbol{U}_1^H\boldsymbol{\Sigma}_1\boldsymbol{U}_1$.
\IF{$\boldsymbol{\Sigma}_1$ is a rank-one matrix}
\STATE $\bar{\boldsymbol{\phi}}^{\star}=\boldsymbol{U}_1[1,:]\sqrt{\boldsymbol{\Sigma}_1[1,1]}$;
\ELSE
\STATE Initialize $\cal{D}=\varnothing$.
\FOR{$d=1 \to D$}
\STATE Generate random vectors $\boldsymbol{\phi}_d=\boldsymbol{U}_1^H \boldsymbol{\Sigma}_1^{\frac12}\boldsymbol{e}_d$, where $\boldsymbol{e}_d \sim \mathcal{CN}(0,\boldsymbol{I}_N)$.
\IF{$\boldsymbol{\phi}_d$ satisfies the constraint of the Problem $\bf (Ps-5)$}
\STATE ${\cal{D}}={\cal{D}}\cup {\boldsymbol{\phi}_d}$.\\
\STATE Obtain the objective function value as $Q_d$.
\ENDIF
\ENDFOR
\STATE $\bar{\boldsymbol{\phi}}^{\star}=\arg \underset{ d \in \cal{D}} {\max} \  Q_d$.
\ENDIF
\end{algorithmic}
\end{algorithm}

\begin{algorithm}[t!]
\caption{ SUM algorithm to solve $\bf (Ps-1)$}\label{Algorithm1_single}
\begin{algorithmic}[1]
\STATE Solve $\bf (Ps-5)$ to obtain $\boldsymbol{\Psi}$.
\STATE Use Algorithm 1 to perform Gaussian randomization procedure for $\boldsymbol{\Psi}$ and obtain new $\bar{\boldsymbol{\phi}}$.
\STATE Obtain $\boldsymbol{\phi}^{\star}$ from (\ref{gainphi*}).
\STATE Obtain $\tau^{\star}$ from (\ref{obtain_tau}).
\STATE Return $\boldsymbol{\phi}^{\star}$ and $\tau^{\star}$.
\end{algorithmic}
\end{algorithm}
\subsection{Genetic-SNR-Maximization (GSM) Algorithm}
RIS phase shift matrix $\boldsymbol{\Phi}$ can be obtained from multiple candidates that maximize the minimum SNRs of the two users, which is a genetic algorithm. In the $i$th candidate generation, $\boldsymbol{\Phi}^{(i-1)}$ represents the RIS phase shifts obtained in the previous generation and the denominators of SNRs in (\ref{gamma_s1}) are approximated by using the previously generated phase shift matrix as follows:
\begin{align}
    \eta_1^{(i)}\approx1+\frac{P_S}{P_B}+\frac{P_S|h_2+\boldsymbol{v}\boldsymbol{\Phi}^{(i-1)}\boldsymbol{g}_2|^2+\sigma^2}{P_B|h_1+\boldsymbol{v}\boldsymbol{\Phi}^{(i-1)}\boldsymbol{g}_1|^2},\label{eta1}
\end{align}
\begin{align}
    \eta_2^{(i)}\approx1+\frac{P_S}{P_B}+\frac{P_S|h_1+\boldsymbol{v}\boldsymbol{\Phi}^{(i-1)}\boldsymbol{g}_1|^2+\sigma^2}{P_B|h_2+\boldsymbol{v}\boldsymbol{\Phi}^{(i-1)}\boldsymbol{g}_2|^2}\label{eta2}.
\end{align}
The $i$th genetic SNR maximization problem is then formulated as follows:
\begin{subequations}
\label{eq:Ps-6}
\begin{align}
{\bf (Ps-6) }: \underset{\boldsymbol{\Phi}^{(i)}}{\max}  \quad  &\min\left\{\frac{|h_1+\boldsymbol{v\Phi}^{(i)}\boldsymbol{g}_1|^2}{\eta_2^{(i)}},\frac{|h_2+\boldsymbol{v\Phi}^{(i)}\boldsymbol{g}_2|^2}{\eta_1^{(i)}}\right\}\nonumber \\
\text{s.t.}\quad
&|\phi_n|=1,\forall n.\tag{33}
\end{align}
\end{subequations}
The above problem can be recast as
\begin{subequations}
\label{eq:Ps-7}
\begin{align}
{\bf (Ps-7) }: \underset{\boldsymbol{\Psi}^{(i)},t}{\max}  \quad  &t\nonumber \\
\text{s.t.}\quad
&\boldsymbol{\Psi}^{(i)}[n,n]=1,\forall n, \\
&\boldsymbol{\Psi}^{(i)} \succeq 0,\\
&\text{tr}[\boldsymbol{\Psi}^{(i)}\bar{\boldsymbol{g}}_1\bar{\boldsymbol{g}}_1^H] \geq \eta_2^{(i)}t,\\
&\text{tr}[\boldsymbol{\Psi}^{(i)}\bar{\boldsymbol{g}}_2\bar{\boldsymbol{g}}_2^H] \geq \eta_1^{(i)}t.
\end{align}
\end{subequations}
This problem can be solved via CVX. In the GSM algorithm, we first initialize the RIS phase shifts $\boldsymbol{\phi}$ and $\tau$ with the solution obtained in Algorithm 2. We then calculate $\eta_1^{(i)}$ and $\eta_2^{(i)}$ with the solution obtained in the previous generation and solve a problem $\bf (P7-S)$ to obtain updated $\boldsymbol{\Psi}$. A rank-one solution is found by following Algorithm 1 and $\boldsymbol{\phi}^{\star}$ is given by (\ref{gainphi*}). We repeat the generation steps and record the minimum SNR until the fixed number of generations is reached. Finally, after the candidate generation, we choose the $\boldsymbol{\phi}^{\star}$ that maximizes the minimum SNR as the optimal solution. The overall GSM algorithm is summarized in Algorithm 3.
\begin{algorithm}[t!]
\caption{GSM algorithm to solve $\bf (Ps-1)$}\label{Algorithm2_single}
\begin{algorithmic}[1]
\STATE Initialize $\boldsymbol{\Phi}$ with the solution obtained in Algorithm 2. $I$ is the candidate generation number.
\FOR{$i=1 \to I$}
\STATE Calculate $\eta_1^{(i)}$ and $\eta_2^{(i)}$ by (\ref{eta1}) and (\ref{eta2}).
\STATE Solve $\bf (Ps-7)$ to update $\boldsymbol{\Psi}^{(i)}$.
\STATE Use Algorithm 1 to perform Gaussian randomization procedure for $\boldsymbol{\Psi}^{(i)}$ and obtain new $\bar{\boldsymbol{\phi}}^{(i)}$.
\STATE Obtain $\boldsymbol{\phi}^{(i)}$ from (\ref{gainphi*}).
\STATE Obtain $\tau^{(i)}$ from (\ref{obtain_tau}).
\STATE Record the minimum SNR and the corresponding $\boldsymbol{\phi}^{(i)}$ and $\tau^{(i)}$.
\ENDFOR
\STATE Choose $\boldsymbol{\phi}^{\star}$ and $\tau^{\star}$ that maximize the minimum SNR.
\STATE Return $\boldsymbol{\phi}^{\star}$ and $\tau^{\star}$.
\end{algorithmic}
\end{algorithm}
\subsection{Complexity Analysis}
In the SUM algorithm, the problem $\bf (Ps-5)$ is solved by using CVX and Gaussian randomization procedure. In the GSM algorithm, the problem $\bf (Ps-7)$ is continuously solved until the maximal generation number is reached. Therefore, following the complexity analysis of a typical interior-point method like a primal-dual path-following method\cite{helmberg1996interior}, the complexity of the SUM algorithm is  $\mathcal{O}((N+3)^4(N+1)^{\frac{1}{2}}\text{log}(\frac{1}{\epsilon_s}))$, whereas that of the GSM algorithm is $\mathcal{O}((1+I)(N+3)^4(N+1)^{\frac{1}{2}}\text{log}(\frac{1}{\epsilon_s}))$, where $\epsilon_s$ is the predefined solution accuracy and $I$ is the candidate generation number in the GSM algorithm.
\begin{remark}
The SUM algorithm is a one-step algorithm, whereas the GSM algorithm is a genetic algorithm that repeatedly generates the candidates of the solution and initialized with the solution obtained in the SUM algorithm. Genetic algorithms are commonly used to generate high-quality solutions to optimization and search problems by relying on biologically inspired operations, such as mutation, crossover, and selection. In practical scenarios, the candidate generation number is not large. By setting different candidate generation numbers, we are able to balance the complexity and the performance. Specifically, to obtain better performance, we can set the candidate generation number large to allow further exploration in the genetic algorithms. When the phase shift and beamforming matrices have to be updated frequently, we may set the candidate generation number small to reduce the complexity. The SUM algorithm is much simpler, whereas the GSM algorithm can achieve better performance at the cost of computational complexity. By proposing the SUM and GSM algorithms, we aim to achieve a favorable tradeoff between performance and complexity.
\end{remark}
\begin{remark}
From the optimization for the single-antenna BS case case, we can obtain insight into how to maximize the upper bound of the SNR to obtain the optimal phase shift matrix. Furthermore, we can obtain the optimal power amplification parameter in a closed-form under the single-antenna BS case. On the other hand, for the multiple-antenna BS case, we need to optimize both beamforming and phase shift matrices.
\end{remark}

\section{RIS for Multiple-Antenna BS}\label{secproposed1}
In this section, we propose algorithms to design a phase shift matrix of RIS and a beamforming matrix of multiple-antenna BS. To solve the original problem ${\bf (Po)}$ in (\ref{eq:Po}), it is divided into two subproblems, namely one for the RIS phase shift matrix and the other one for the BS beamforming matrix. Following the similar optimization procedure in Section III, an SUM-OB algorithm is devised. We then provide a low complexity method to obtain the beamforming matrix by utilizing a maximal-ratio-reception maximal-ratio-transmission (MRR-MRT) relaying scheme in \cite{liang2008optimal} and propose the SUM-MRB algorithm. The GSM-OB and GSM-MRB algorithms are also developed to obtain the phase shifts $\boldsymbol{\Phi}$ and the beamforming matrix $\boldsymbol{A}$ by using a genetic algorithm for the case of multiple-antenna BS in TWRN.
\subsection{SUM-OB Algorithm}
\subsubsection{Phase Shift Optimization}
From (A1), the SNR upper bounds of S1 and S2 are $\beta\left\|\boldsymbol{h}_2+\boldsymbol{V\Phi}\boldsymbol{g}_2\right\|^2$ and $\beta\left\|\boldsymbol{h}_1+\boldsymbol{V\Phi}\boldsymbol{g}_1\right\|^2$, respectively. Similarly to the single-antenna BS case, an optimization problem to find the optimal phase shift matrix that can maximize the upper bound of minimum SNR of users for the multiple-antenna BS case can be formulated as follows:
\begin{subequations}
\label{eq:Pm-a1}
\begin{align}
{\bf (Pm-a1) }: \underset{\boldsymbol{\Phi}}{\max}  \quad  &\min\left\{\left\|\boldsymbol{h}_1+\boldsymbol{V\Phi}\boldsymbol{g}_1\right\|^2,\left\|\boldsymbol{h}_2+\boldsymbol{V\Phi}\boldsymbol{g}_2\right\|^2\right\} \nonumber \\
\text{s.t.}\quad
&|\phi_n|=1,\forall n. \tag{35}
\end{align}
\end{subequations}
We note that the maximization of the upper bound of SNR conveys physical meaning. The more the combined channel gain is, the larger the SNR becomes.
By introducing an additional variable $Q$ and replacing the variable $\boldsymbol{\Phi}$ with $\bar{\boldsymbol{\phi}}$, the problem is further transformed to the following rank-relaxed SDP problem:
\begin{subequations}
\label{eq:Pm-a2}
\begin{align}
{\bf (Pm-a2) }: \underset{\boldsymbol{\Psi},Q}{\max}  \quad  &Q\nonumber \\
\text{s.t.}\quad
&\boldsymbol{\Psi}[n,n]=1,\forall n, \\
&\boldsymbol{\Psi} \succeq 0,\\
&\text{tr}[\boldsymbol{\Psi}\bar{\boldsymbol{G}}_1^H\bar{\boldsymbol{G}}_1] \geq Q,\\
&\text{tr}[\boldsymbol{\Psi}\bar{\boldsymbol{G}}_2^H\bar{\boldsymbol{G}}_2] \geq Q.
\end{align}
\end{subequations}
The SDP problem ${\bf (Pm-a2)}$ can be efficiently solved via CVX. From the optimal $\boldsymbol{\Psi}$ of (\ref{eq:Pm-a2}), optimal rank-one solution $\boldsymbol{\phi}^{\star}$ is obtained from Algorithm 1 and (\ref{gainphi*}).

\subsubsection{Beamforming Matrix Optimization}
After obtaining phase shift matrix $\boldsymbol{\Phi}$ for RIS, the original problem $\bf (Po)$ is transformed to
\begin{align}
{\bf (Pm-b1) }: \underset{\boldsymbol{A}}{\max}  \quad  &\min\left\{\gamma_1,\gamma_2\right\} \notag\\
\text{s.t.}\quad
&(\ref{eqp4}),(\ref{eqp5}), \text{and} \ (\text{\ref{eqp3}})\notag.
\end{align}
Defining $\widetilde{\boldsymbol{h}}_1\triangleq \boldsymbol{h}_1+\boldsymbol{V\Phi}\boldsymbol{g}_1\in \mathbb{C}^{M\times 1}$ and $\widetilde{\boldsymbol{h}}_2\triangleq \boldsymbol{h}_2+\boldsymbol{V\Phi}\boldsymbol{g}_2\in \mathbb{C}^{M\times 1}$, problem ${\bf (Pm-b1)}$ is further rewritten as follows:
\begin{subequations}
\label{eq:Pm-b2}
\begin{align}
{\bf (Pm-b2) }: \underset{\boldsymbol{A}}{\max}  \quad  &\min\left\{\frac{\left|\widetilde{\boldsymbol{h}}_1^T\boldsymbol{A}\widetilde{\boldsymbol{h}}_2\right|^2}{\left\|\widetilde{\boldsymbol{h}}_1^T\boldsymbol{A}\right\|^2+1},\frac{\left|\widetilde{\boldsymbol{h}}_2^T\boldsymbol{A}\widetilde{\boldsymbol{h}}_1\right|^2}{\left\|\widetilde{\boldsymbol{h}}_2^T\boldsymbol{A}\right\|^2+1}\right\}\nonumber \\
\text{s.t.}\quad
&   P_S\left\|\boldsymbol{A}\widetilde{\boldsymbol{h}}_1\right\|^2+P_S\left\|\boldsymbol{A}\widetilde{\boldsymbol{h}}_2\right\|^2+\sigma^2 \text{tr}[\boldsymbol{A}\boldsymbol{A}^H]\leq P_B.\label{eqp10}\tag{37}
\end{align}
\end{subequations}
By introducing an additional variable $Q'$, problem ${\bf (Pm-b2)}$ is recast as
\begin{subequations}
\label{eq:Pm-b3}
\begin{align}
{\bf (Pm-b3) }: \underset{\boldsymbol{A},Q'}{\max}  \quad  &Q'\nonumber \\
\text{s.t.}\quad
&\frac{\left|\widetilde{\boldsymbol{h}}_1^T\boldsymbol{A}\widetilde{\boldsymbol{h}}_2\right|^2}{\left\|\widetilde{\boldsymbol{h}}_1^T\boldsymbol{A}\right\|^2+1} \geq Q',\label{m-b3-first}\\
&\frac{\left|\widetilde{\boldsymbol{h}}_2^T\boldsymbol{A}\widetilde{\boldsymbol{h}}_1\right|^2}{\left\|\widetilde{\boldsymbol{h}}_2^T\boldsymbol{A}\right\|^2+1} \geq Q',\label{m-b3-second}\\
&(\ref{eqp10}).\nonumber
\end{align}
\end{subequations}
By denoting $\text{vec}\left(\boldsymbol{A}\right)=\boldsymbol{a}\in \mathbb{C}^{M^2\times 1}$, we have
\begin{align}
\left|\widetilde{\boldsymbol{h}}_1^T\boldsymbol{A}\widetilde{\boldsymbol{h}}_2\right|^2=\ &\text{tr}\left[(\widetilde{\boldsymbol{h}}_1^T\boldsymbol{A}\widetilde{\boldsymbol{h}}_2)^H \widetilde{\boldsymbol{h}}_1^T\boldsymbol{A}\widetilde{\boldsymbol{h}}_2\right]=\text{vec}\left(\widetilde{\boldsymbol{h}}_1^T\boldsymbol{A}\widetilde{\boldsymbol{h}}_2\right)^H \text{vec}\left(\widetilde{\boldsymbol{h}}_1^T\boldsymbol{A}\widetilde{\boldsymbol{h}}_2\right) \nonumber\\
=\ &\text{vec}\left(\boldsymbol{A}\right)^H\left(\widetilde{\boldsymbol{h}}_2^T \otimes \widetilde{\boldsymbol{h}}_1^T\right)^H\left(\widetilde{\boldsymbol{h}}_2^T \otimes \widetilde{\boldsymbol{h}}_1^T\right)\text{vec}\left(\boldsymbol{A}\right)=\boldsymbol{a}^H\boldsymbol{C}_1\boldsymbol{a},\\
\left\|\widetilde{\boldsymbol{h}}_1^T\boldsymbol{A}\right\|^2=\ &\text{vec}(\widetilde{\boldsymbol{h}}_1^T\boldsymbol{A})^H \text{vec}(\widetilde{\boldsymbol{h}}_1^T\boldsymbol{A})=\text{vec}\left(\boldsymbol{A}\right)^H\left(\boldsymbol{I} \otimes \widetilde{\boldsymbol{h}}_1^T\right)^H\left(\boldsymbol{I} \otimes \widetilde{\boldsymbol{h}}_1^T\right)\text{vec}\left(\boldsymbol{A}\right)\nonumber\\
=\ &\boldsymbol{a}^H\boldsymbol{D}_1\boldsymbol{a}.
\end{align}
Therefore, the first constraint (\ref{m-b3-first}) can be rewritten as
\begin{align}
\boldsymbol{a}^H\left(\boldsymbol{C}_1- Q'\boldsymbol{D}_1\right)\boldsymbol{a}\geq Q'.
\end{align}
After applying the similar techniques to the other two constraints, (\ref{eqp10}) and (\ref{m-b3-second}), the problem ${\bf (Pm-b3)}$ in (\ref{eq:Pm-b3}) is transformed as follows:
\begin{subequations}
\label{eq:Pm-b4}
\begin{align}
{\bf (Pm-b4) }: \underset{\boldsymbol{a},Q'}{\max}  \quad  &Q'\nonumber \\
\text{s.t.}\quad
&\boldsymbol{a}^H\left(\boldsymbol{C}_1-Q'\boldsymbol{D}_1\right)\boldsymbol{a}\geq Q',\\
&\boldsymbol{a}^H\left(\boldsymbol{C}_2-Q'\boldsymbol{D}_2\right)\boldsymbol{a}\geq Q',\\
&\boldsymbol{a}^H\boldsymbol{F}\boldsymbol{a}\leq P_B,
\end{align}
\end{subequations}
where
\begin{align}
\boldsymbol{C}_1=&\left(\widetilde{\boldsymbol{h}}_2^T \otimes \widetilde{\boldsymbol{h}}_1^T\right)^H\left(\widetilde{\boldsymbol{h}}_2^T \otimes \widetilde{\boldsymbol{h}}_1^T\right),
\boldsymbol{D}_1=\left(\boldsymbol{I} \otimes \widetilde{\boldsymbol{h}}_1^T\right)^H\left(\boldsymbol{I} \otimes \widetilde{\boldsymbol{h}}_1^T\right),\\
\boldsymbol{C}_2=&\left(\widetilde{\boldsymbol{h}}_1^T \otimes \widetilde{\boldsymbol{h}}_2^T\right)^H\left(\widetilde{\boldsymbol{h}}_1^T \otimes \widetilde{\boldsymbol{h}}_2^T\right),
\boldsymbol{D}_2=\left(\boldsymbol{I} \otimes \widetilde{\boldsymbol{h}}_2^T\right)^H\left(\boldsymbol{I} \otimes \widetilde{\boldsymbol{h}}_2^T\right),\\
\boldsymbol{F}=&P_S\left(\widetilde{\boldsymbol{h}}_1^T\otimes \boldsymbol{I}\right)^H\left(\widetilde{\boldsymbol{h}}_1^T\otimes \boldsymbol{I}\right)+P_S\left(\widetilde{\boldsymbol{h}}_2^T\otimes \boldsymbol{I}\right)^H\left(\widetilde{\boldsymbol{h}}_2^T\otimes \boldsymbol{I}\right)+\sigma^2 \boldsymbol{I}\label{eqp-12-5}.
\end{align}
Defining $\boldsymbol{\Xi}\triangleq\boldsymbol{a}\boldsymbol{a}^H\in \mathbb{C}^{M^2\times M^2}$ and relaxing the rank-one constraint for $\boldsymbol{\Xi}$, optimization problem ${\bf (Pm-b4)}$ can be recast as follows:
\begin{subequations}
\label{eq:Pm-b5}
\begin{align}
{\bf (Pm-b5) }: \underset{\boldsymbol{\Xi},Q'}{\max}  \quad  &Q'\nonumber \\
\text{s.t.}\quad
&\boldsymbol{\Xi} \succeq 0,\\
&\text{tr}\left[\boldsymbol{\Xi}\left(\boldsymbol{C}_1-Q'\boldsymbol{D}_1\right)\right]\geq Q',\\
&\text{tr}\left[\boldsymbol{\Xi}\left(\boldsymbol{C}_2-Q'\boldsymbol{D}_2\right)\right]\geq Q',\\
&\text{tr}\left[\boldsymbol{\Xi}\boldsymbol{F}\right]\leq P_B.
\end{align}
\end{subequations}

This problem is still non-convex because $Q'$ and $\boldsymbol{\Xi}$ are coupled, yet it is convex with respect to each of $Q'$ and $\boldsymbol{\Xi}$. Thus, we find the solution by using a convex feasibility problem test\cite{joung2014energy}. First, we obtain $Q_{low}$ and $Q_{up}$ of $Q'$ where $Q_{low}$ makes the problem $\bf (Pm-b5)$ feasible and $Q_{up}$ makes the problem $\bf (Pm-b5)$ infeasible. We then calculate $Q_{new}=\frac{Q_{up}+Q_{low}}{2}$ and test the feasibility of the problem $\bf (Pm-b5)$ by replacing $Q'$ with $Q_{new}$. If $Q_{new}$ makes the problem feasible, we update $Q_{low}=Q_{new}$; otherwise, we update $Q_{up}=Q_{new}$. These steps are repeated until a stopping criterion is met. For a certain $Q'$, the problem is an SDP problem. $\boldsymbol{\Xi}$ can be obtained by solving the feasibility problem via CVX. With $\boldsymbol{\Xi}$, Algorithm 1 performs the Gaussian randomization procedure to obtain $\boldsymbol{a}$, and the beamforming matrix $\boldsymbol{A}$ can be recovered from $\boldsymbol{a}$ by reshaping $\boldsymbol{a}$ as an $M \times M$ matrix. The detailed steps of the SUM-OB algorithm are summarized in Algorithm 4.

\begin{algorithm}[t!]
\caption{ SUM-OB algorithm to solve $\bf (Po)$}\label{Algorithm1_multi}
\begin{algorithmic}[1]
\STATE Solve $\bf (Pm-a2)$ to obtain $\boldsymbol{\Psi}^{\star}$.
\STATE Use Algorithm 1 to perform Gaussian randomization procedure for $\boldsymbol{\Psi}^{\star}$ and obtain  $\bar{\boldsymbol{\phi}}^{\star}$.
\STATE Compute $\boldsymbol{\phi}^{\star}$ by (\ref{gainphi*}).
\STATE For given $\boldsymbol{\phi}^{\star}$ and $Q_{low}, Q_{up}$, \WHILE{$Q_{up}-Q_{low}\geq \epsilon_1$}
\STATE Calculate $Q_{new}=\frac{Q_{up}+Q_{low}}{2}$.
\STATE{Solve the feasibility problem $\bf (Pm-b5)$ with given $Q'=Q_{new}$.}
\IF {the problem is feasible}
\STATE $Q_{low}=Q_{new}$, update $\boldsymbol{\Xi}$.
\ELSE
\STATE $Q_{up}=Q_{new}$.
\ENDIF
\ENDWHILE
\STATE Obtain $\boldsymbol{\Xi}^{\star}$.
\STATE Use Algorithm 1 to perform Gaussian randomization procedure for $\boldsymbol{\Xi}^{\star}$ and obtain $\boldsymbol{a}$.
\STATE Recover $\boldsymbol{A}^{\star}$ from $\boldsymbol{a}$.
\STATE  Return $\boldsymbol{\phi}^{\star}$ and $\boldsymbol{A}^{\star}$.
\end{algorithmic}
\end{algorithm}

\subsection{SUM-MRB Algorithm}
Since the computational complexity to obtain the beamforming matrix $\boldsymbol{A}$ is high, a low-complexity SUM-MRB algorithm is devised. In the SUM-MRB algorithm, RIS phase shift matrix $\boldsymbol{\Phi}$ is obtained through the same method as the SUM-OB algorithm, whereas beamforming matrix $\boldsymbol{A}$ is obtained through an MRR-MRT relaying scheme in \cite{liang2008optimal}.

Defining $\boldsymbol{H}_1\triangleq\left[\boldsymbol{h}_1+\boldsymbol{V\Phi}\boldsymbol{g}_1,\boldsymbol{h}_2+\boldsymbol{V\Phi}\boldsymbol{g}_2\right]\in \mathbb{C}^{M\times 2}$ and $\boldsymbol{H}_2\triangleq\left[\boldsymbol{h}_2+\boldsymbol{V\Phi}\boldsymbol{g}_2,\boldsymbol{h}_1+\boldsymbol{V\Phi}\boldsymbol{g}_1\right]^T\in \mathbb{C}^{2\times M}$, beamforming matrix $\boldsymbol{A}$ can be designed as follows \cite{liang2008optimal}:
\begin{align}
\boldsymbol{A}=&\alpha \boldsymbol{H}_2^H \boldsymbol{H}_1^H=\alpha \left(\bar{\boldsymbol{G}}_2\bar{\boldsymbol{\phi}}\bar{\boldsymbol{\phi}}^T\bar{\boldsymbol{G}}_1^T+\bar{\boldsymbol{G}}_1\bar{\boldsymbol{\phi}}\bar{\boldsymbol{\phi}}^T\bar{\boldsymbol{G}}_2^T\right)^*=\alpha \boldsymbol{A}_1,\label{A_structure}
\end{align}
where $\alpha$ is used to satisfy the power constraint with the equality in (\ref{eqp3}). Specifically, $\alpha$ is derived as
\begin{align}
\alpha=\sqrt{\frac{P_B}{P_S\text{tr}\left[\boldsymbol{A}_1\left(\bar{\boldsymbol{G}}_1\boldsymbol{\Psi}\bar{\boldsymbol{G}}_1^H+\bar{\boldsymbol{G}}_2\boldsymbol{\Psi}\bar{\boldsymbol{G}}_2^H\right)\boldsymbol{A}_1^H\right]+\sigma^2\text{tr}\left[\boldsymbol{A}_1\boldsymbol{A}_1^H\right]}}.\label{A_ratio}
\end{align}
In (\ref{A_structure}), $\boldsymbol{H}_1^H$ corresponds to
receive beamforming at BS, which maximizes the received energy along the direction of $\boldsymbol{h}_1+\boldsymbol{V\Phi}\boldsymbol{g}_1$ and $\boldsymbol{h}_2+\boldsymbol{V\Phi}\boldsymbol{g}_2$, and $\boldsymbol{H}_2^H$ corresponds to transmit beamforming, which maximizes the transmission energy along the direction of $(\boldsymbol{h}_2+\boldsymbol{V\Phi}\boldsymbol{g}_2)^T$ from BS to S2 and $(\boldsymbol{h}_1+\boldsymbol{V\Phi}\boldsymbol{g}_1)^T$ from BS to S1.
This scheme is called an MRR-MRT relaying scheme.

 The SUM-MRB algorithm that employs the MRR-MRT to obtain the beamforming matrix is summarized in Algorithm 5. We note that both SUM-OB and the SUM-MRB algorithms are a one-step algorithm. In the next subsection, we will present the genetic algorithms, namely, GSM-OB and GSM-MRB algorithms, for the multiple-antenna BS case.

\begin{algorithm}[t!]
\caption{SUM-MRB algorithm to solve $\bf (Po)$}\label{Algorithm1_multi}
\begin{algorithmic}[1]
\STATE Solve $\bf (Pm-a2)$ to obtain $\boldsymbol{\Psi}^{\star}$.
\STATE Use Algorithm 1 to perform Gaussian randomization procedure for $\boldsymbol{\Psi}^{\star}$ and obtain  $\bar{\boldsymbol{\phi}}^{\star}$.
\STATE Compute $\boldsymbol{\phi}^{\star}$ by (\ref{gainphi*}).
\STATE Compute $\boldsymbol{A}^{\star}$ by (\ref{A_structure}) and (\ref{A_ratio}).
\STATE Return $\boldsymbol{\phi}^{\star}$ and $\boldsymbol{A}^{\star}$.
\end{algorithmic}
\end{algorithm}

\subsection{GSM-OB and GSM-MRB Algorithms}
In the GSM-OB algorithm, we obtain RIS phase shift matrix $\boldsymbol{\Phi}$ and beamforming matrix $\boldsymbol{A}$ by consequently maximizing the SNRs of the two users. In the $i$th generation, denoting the RIS phase shift and beamforming matrices by $\boldsymbol{\Phi}^{(i-1)}$ and $\boldsymbol{A}^{(i-1)}$, respectively, in the previous generation, we obtain $\boldsymbol{\nu}_1^{(i)}\in \mathbb{C}^{1\times M}$, $\boldsymbol{\nu}_2^{(i)}\in \mathbb{C}^{1\times M}$, $\zeta_1^{(i)}$, and $\zeta_2^{(i)}$ as follows:
\begin{align}
\boldsymbol{\nu}_1^{(i)}=&\left(\boldsymbol{h}_1+\boldsymbol{V}\boldsymbol{\Phi}^{(i-1)}\boldsymbol{g}_1\right)^T\boldsymbol{A}^{(i-1)},\
\boldsymbol{\nu}_2^{(i)}=\left(\boldsymbol{h}_2+\boldsymbol{V}\boldsymbol{\Phi}^{(i-1)}\boldsymbol{g}_2\right)^T\boldsymbol{A}^{(i-1)},\label{nu_12}\\
\zeta_1^{(i)}=&\left\|\boldsymbol{\nu}_1^{(i)}\right\|^2+1,\  \zeta_2^{(i)}=\left\|\boldsymbol{\nu}_2^{(i)}\right\|^2+1\label{zeta_12}.
\end{align}
The SNRs for S1 and S2 are written as
\begin{align}
    \gamma_1^{(i)}=\frac{\beta\left|\boldsymbol{\nu}_1^{(i)}\left(\boldsymbol{h}_2+\boldsymbol{V\Phi}^{(i)}\boldsymbol{g}_2\right)\right|^2}{\zeta_1^{(i)}}=\frac{\beta\left|\boldsymbol{\nu}_1^{(i)}\bar{\boldsymbol{G}}_2\bar{\boldsymbol{\phi}}^{(i)}\right|^2}{\zeta_1^{(i)}},\\
    \gamma_2^{(i)}=\frac{\beta\left|\boldsymbol{\nu}_2^{(i)}\left(\boldsymbol{h}_1+\boldsymbol{V\Phi}^{(i)}\boldsymbol{g}_1\right)\right|^2}{\zeta_2^{(i)}}=\frac{\beta\left|\boldsymbol{\nu}_2^{(i)}\bar{\boldsymbol{G}}_1\bar{\boldsymbol{\phi}}^{(i)}\right|^2}{\zeta_2^{(i)}}.
\end{align}

The genetic SNR maximization problem of the $i$th generation is then formulated as follows:
\begin{subequations}
\label{eq:Pm-a3}
\begin{align}
{\bf (Pm-a3) }: \underset{\boldsymbol{\Phi}^{(i)}}{\max}  \quad  &\min\left\{\frac{\left|\boldsymbol{\nu}_1^{(i)}\bar{\boldsymbol{G}}_2\bar{\boldsymbol{\phi}}^{(i)}\right|^2}{\zeta_1^{(i)}},\frac{\left|\boldsymbol{\nu}_2^{(i)}\bar{\boldsymbol{G}}_1\bar{\boldsymbol{\phi}}^{(i)}\right|^2}{\zeta_2^{(i)}}\right\}\nonumber \\
\text{s.t.}\quad
&\left|\phi_n^{(i)}\right|=1,\forall n.\tag{53}
\end{align}
\end{subequations}

The problem can be transformed to the following rank-relaxed SDP problem:
\begin{subequations}
\label{eq:Pm-a4}
\begin{align}
{\bf (Pm-a4) }: \underset{\boldsymbol{\Psi}^{(i)},Q}{\max}  \quad  &Q\nonumber \\
\text{s.t.}\quad
&\boldsymbol{\Psi}^{(i)}[n,n]=1,\forall n, \\
&\boldsymbol{\Psi}^{(i)} \succeq 0,\\
&\text{tr}\left[\boldsymbol{\Psi}^{(i)}\bar{\boldsymbol{G}}_1^H(\boldsymbol{\nu}_2^{(i)})^H\boldsymbol{\nu}_2^{(i)}\bar{\boldsymbol{G}}_1\right] \geq \zeta_2^{(i)}Q,\\
&\text{tr}\left[\boldsymbol{\Psi}^{(i)}\bar{\boldsymbol{G}}_2^H(\boldsymbol{\nu}_1^{(i)})^H\boldsymbol{\nu}_1\bar{\boldsymbol{G}}_2\right] \geq \zeta_1^{(i)}Q.
\end{align}
\end{subequations}
This problem ${\bf (Pm-a4)}$ can be solved by using CVX. In the GSM-OB algorithm, we first initialize RIS phase shift vector $\boldsymbol{\phi}$ and beamforming matrix $\boldsymbol{A}$ with the solution obtained in Algorithm 4. We then calculate $\nu_1^{(i)}$, $\nu_2^{(i)}$, $\zeta_1^{(i)}$, and $\zeta_2^{(i)}$ with the solution obtained in the previous generation, and solve problem $\bf (Pm-a4)$ to update $\boldsymbol{\Psi}$. A rank-one solution is obtained from Algorithm 1 and (\ref{gainphi*}). The update of $\boldsymbol{A}$ is obtained by solving a problem $\bf (Pm-b5)$. We repeat these steps and record the minimum SNR and corresponding solution in each generation. Finally, after completing the predetermined number of generations, we choose the $\boldsymbol{\phi}^{\star}$ and $\boldsymbol{A}^{\star}$ that maximize the minimum SNR as the optimal solution of a GSM-OB algorithm. The detailed steps of the GSM-OB algorithm are summarized in Algorithm 6.

We now describe the GSM-MRB algorithm. In Algorithm 6, i.e., a GSM-OB algorithm, by replacing the seventh step with computing $\boldsymbol{A}$ by (\ref{A_structure}) and (\ref{A_ratio}) and by initializing the starting point in the first step with the solution obtained in the SUM-MRB algorithm, we can implement the GSM-MRB algorithm. Since other steps of the GSM-MRB algorithm are the same as the GSM-OB algorithm, the detailed steps of GSM-MRB algorithm are omitted here.
\begin{algorithm}[t!]
\caption{GSM-OB algorithm to solve $\bf (Po)$}\label{Algorithm2_multi}
\begin{algorithmic}[1]
\STATE Initialize $\boldsymbol{\Phi}$ and $\boldsymbol{A}$ with the solution obtained in Algorithm 4. $I$ is the candidate generation number.
\FOR{$i=1 \to I$}
\STATE Calculate $\boldsymbol{\nu}_1^{(i)}$, $\boldsymbol{\nu}_2^{(i)}$, $\zeta_1^{(i)}$, and $\zeta_2^{(i)}$ by (\ref{nu_12}), (\ref{zeta_12}).
\STATE Solve $\bf (Pm-a4)$ to update $\boldsymbol{\Psi}^{(i)}$.
\STATE Use Algorithm 1 to perform Gaussian randomization procedure for $\boldsymbol{\Psi}^{(i)}$ and update $\bar{\boldsymbol{\phi}}^{(i)}$.
\STATE Obtain $\boldsymbol{\phi}^{(i)}$ from (\ref{gainphi*}).
\STATE Solve $\bf (Pm-b5)$ to update $\boldsymbol{\Xi}^{(i)}$.
\STATE Use Algorithm 1 to perform Gaussian randomization procedure for $\boldsymbol{\Xi}^{(i)}$ and update $\boldsymbol{a}^{(i)}$.
\STATE Recover $\boldsymbol{A}^{(i)}$ from $\boldsymbol{a}^{(i)}$.
\STATE Record the minimum SNR and the corresponding $\boldsymbol{\phi}^{(i)}$ and $\boldsymbol{A}^{(i)}$.
\ENDFOR
\STATE Choose $\boldsymbol{\phi}^{\star}$ and $\boldsymbol{A}^{\star}$ that maximize the minimum SNR.
\STATE Return $\boldsymbol{\phi}^{\star}$ and $\boldsymbol{A}^{\star}$.
\end{algorithmic}
\end{algorithm}

\subsection{Complexity Analysis}
In the SUM-OB algorithm, sub-problem $\bf (Pm-a2)$ is solved and then the sub-problem $\bf (P5-b)$ is solved alternately. To be specific, the iteration number is $\text{log}_2(\frac{Q_{up}-Q_{low}}{\epsilon_1})$ to achieve an accuracy of $\epsilon_1$ for $Q$. Hence, We can obtain the complexity order of the SUM-OB algorithm as $\mathcal{O}\bigg((N+3)^4(N+1)^{\frac{1}{2}}\text{log}(\frac{1}{\epsilon_s})+\text{log}_2(\frac{Q_{up}-Q_{low}}{\epsilon_1})M^9log(\frac{1}{\epsilon_s})\bigg)$. The SUM-MRB algorithm obtains the beamforming matrix with less complexity and the complexity is given by $\mathcal{O}\bigg((N+3)^4(N+1)^{\frac{1}{2}}\text{log}(\frac{1}{\epsilon_s})+2M^3+4M^2(N+1)\allowbreak+2M(N+1)^2\bigg)$. The GSM-OB and GSM-MRB algorithms are genetic algorithms and their computational complexity orders are analyzed as $\mathcal{O}\bigg((1+I)\big((N+3)^4\allowbreak (N+1)^{\frac{1}{2}}\text{log}(\frac{1}{\epsilon_s})+\text{log}_2(\frac{Q_{up}-Q_{low}}{\epsilon_1})M^9log(\frac{1}{\epsilon_s})\big)\bigg)$ and $\mathcal{O}\bigg((1+I)\big((N+3)^4(N+1)^{\frac{1}{2}}\text{log}(\frac{1}{\epsilon_s})+2M^3\allowbreak+4M^2(N+1)+2M(N+1)^2\big)\bigg)$, respectively.
\begin{remark}
For the multiple-antenna BS case, we propose four algorithms, namely the SUM-OB, SUM-MRB, GSM-OB, and GSM-MRB algorithms. Their relationships can be summarized in Table \ref{table_Proposed algorithms}. The SUM-OB and SUM-MRB algorithms are the one-step algorithms, whereas the GSM-OB and GSM-MRB algorithms are the genetic algorithms. From the complexity analysis, we can conclude that the SUM-MRB algorithm requires the least computational complexity, whereas the GSM-OB algorithm can achieve the best performance, as numerically verified in the next section. There exists a performance-and-complexity tradeoff between the proposed algorithms.
\end{remark}

\section{Simulation Results}\label{secsimulation}
In this section, simulation results are provided to validate the effectiveness of the proposed algorithms. In the simulations, we consider a TWRN, in which the BS is equipped with four antennas and the RIS is a uniform rectangular array with $10 \times 10$ reflective elements. We set $d_{1,\text{R}}=40\  \text{m}$, $d_{2,\text{R}}=60\  \text{m}$, $d_{\text{B,R}}=80\  \text{m}$, $d_{\text{B},1}=\sqrt{d_{1,\text{R}}^2+d_{\text{B,R}}^2}$ and $d_{\text{B},2}=\sqrt{d_{2,\text{R}}^2+d_{\text{B,R}}^2}$.
The system parameters are nearly identical to those in \cite{kammoun2020asymptotic}. Specifically, the path loss is set according to the 3GPP Urban Micro (UMi) scenario from \cite{access2010further} with a carrier frequency of 2.5 GHz. The path loss is set as follows:
\begin{align}
\beta(d) [\text{dB}]=\left\{
\begin{array}{rcl}
G_t+G_r-35.95-22\text{log}_{10}(d),       &      & \text{for an LoS channel,}\\
G_t+G_r-33.05-36.7\text{log}_{10}(d),       &      & \text{for an NLoS channel,}
\end{array} \right.
\end{align}
where $G_t$ and $G_r$ denote the corresponding antenna gains (in dBi) at the transmitter and receiver, respectively. We assume that the BS and RIS have a gain of 5 dBi and the users have a gain of 0 dBi. The bandwidth $B=180$ kHz, and the noise power  $\sigma^2=-174+10\text{log}_{10}(B) $ dBm. We set the Rician factor $K_v=K_1=K_2=10$. The center azimuth AoA is chosen randomly from $[-\pi,\pi]$, whereas the elevation AoA is chosen randomly from $[-25^\circ,25^\circ]$. Simulation results are based on $10^3$ random channel realizations.
\subsection{RIS for Single-Antenna BS}
The minimum SNRs of the proposed algorithms are evaluated when $P_S=0$ dBm by varying the transmit power budget of BS with a single antenna. For comparison purpose, the following two benchmark schemes are considered.
\begin{itemize}
\item Benchmark 1 (RIS: random phase): The phase shifts are randomly chosen from $[0,2\pi]$ and the power amplification parameter $\tau$ with which the BS amplifies its received signal is given by (\ref{obtain_tau}).
\item Benchmark 2 (No RIS): A TWR system without RIS is considered and the power amplification parameter $\tau$ is given by $\tau=\frac{P_B}{P_S|h_1|^2+P_S|h_2|^2+\sigma^2}$.
\end{itemize}

\begin{figure}[t]
 \centering
 \begin{minipage}[t]{0.5\textwidth}
  \centering
  \includegraphics[width=3in]{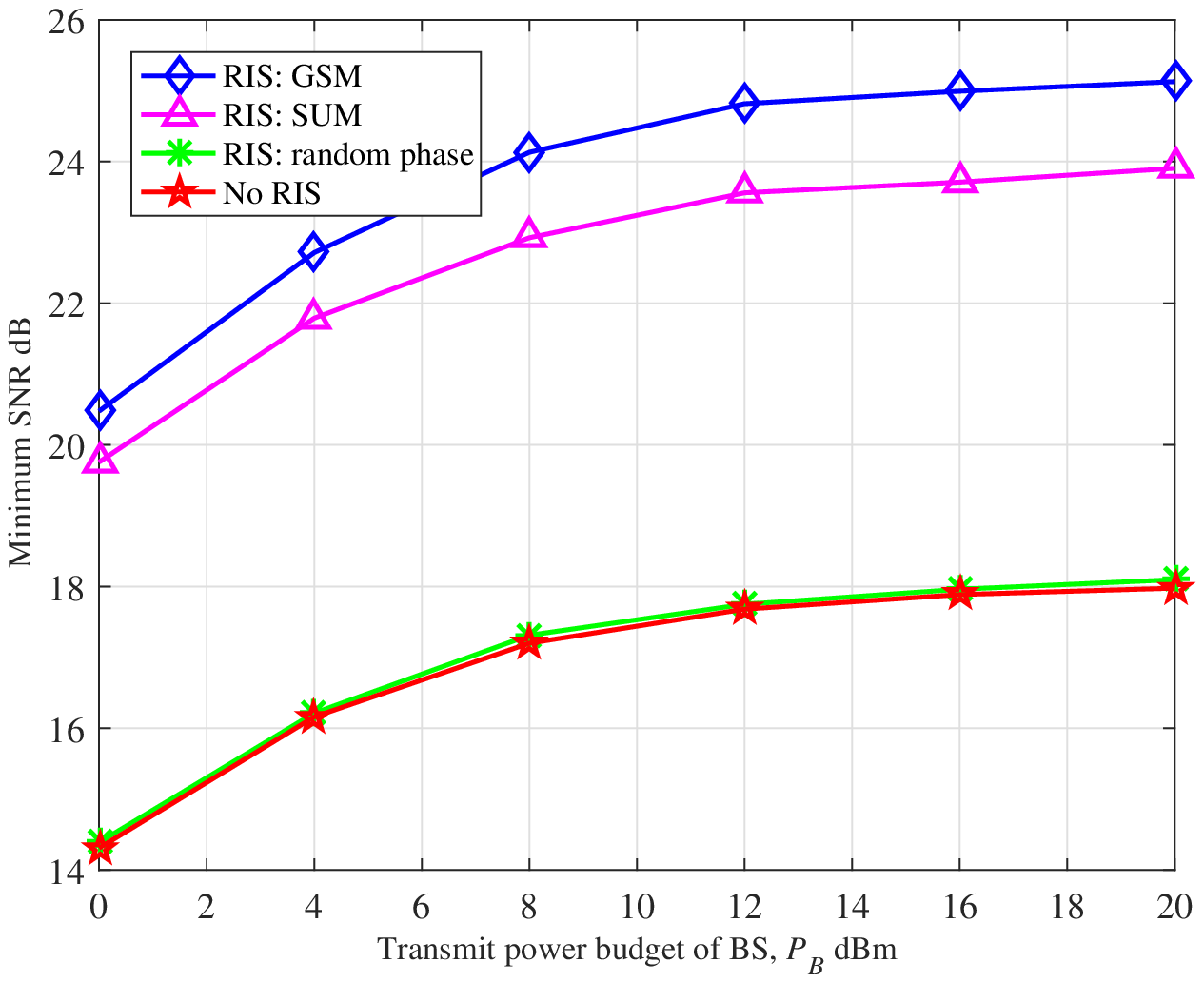}
  \caption{Minimum SNR over transmit power budget of single-antenna BS when $P_S=0$ dBm.}
  \label{fig:Fig2}
 \end{minipage}
 \begin{minipage}[t]{0.48\textwidth}
  \centering
  \includegraphics[width=3in]{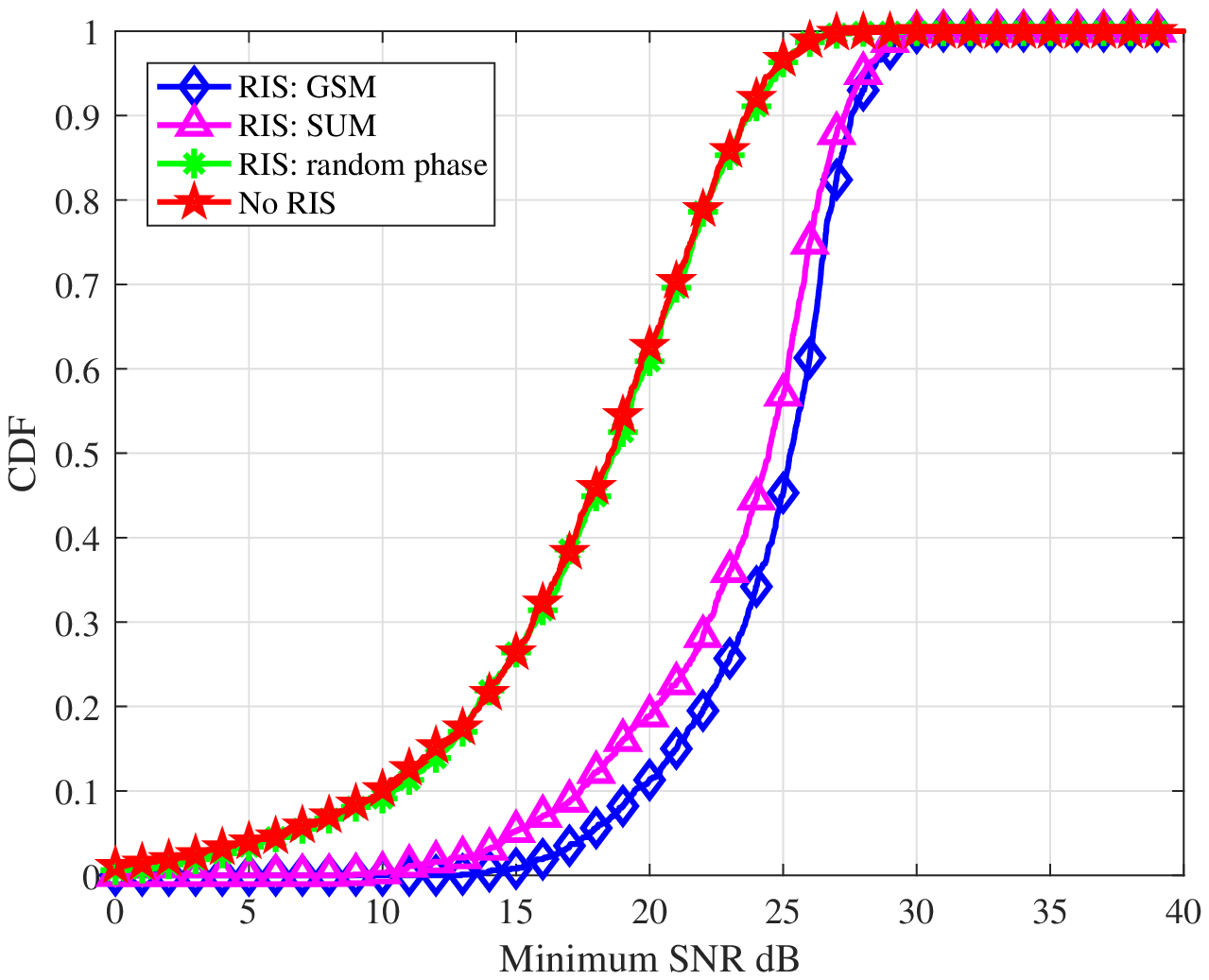}
  \caption{CDF of minimum SNR for different algorithms when $P_B=10$ dBm and $P_S=0$ dBm.}
  \label{fig:Fig3}
 \end{minipage}
\end{figure}

\begin{figure}[t]
\centering
\includegraphics[width=3in] {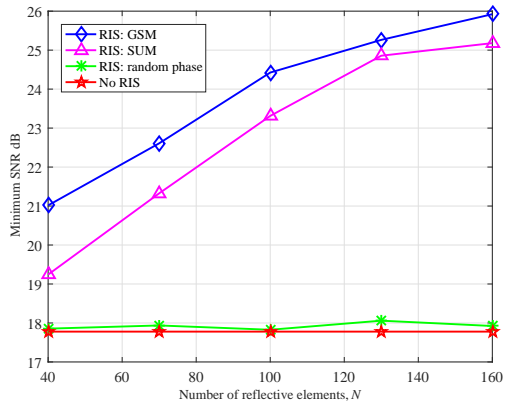}
\caption{Minimum SNR over the number of the reflective elements for the single-antenna BS case when $P_B=10$ dBm and $P_S=0$ dBm.}
\label{fig:Fig4}
\end{figure}

Fig. \ref{fig:Fig2} shows users' minimum SNR versus the transmit power budget for different algorithms. It is observed that the minimum SNR increases and is saturated as the transmit power budget increases. The minimum SNR saturation is owing to the extant additive noise. Besides, we observe that the performance gain by deploying the RIS is negligible if the phase shifts are randomly chosen. Also, the proposed SUM and GSM algorithms can achieve approximately 5 dB and 6 dB gain, respectively, compared to the benchmark schemes. As expected, the GSM algorithm outperforms the SUM algorithm at the cost of the computational complexity.

Fig. \ref{fig:Fig3} shows the \emph{cumulative distribution function} (CDF) of the minimum SNR for different algorithms when $P_B=10$ dBm and $P_S=0$ dBm. We can take the CDF as the user success probability and the minimum SNR as the users' target SNR. Here, the users' target SNR means both users can decode their information correctly if their real SNR is higher than this target SNR. The user success probability is referred to as the probability that the users can meet the SNR constraints. When the minimum SNR takes the value of 20 dB, for example, the corresponding CDF value is 0.114, 0.188, 0.614, and 0.625 for the GSM, SUM, random phase RIS, and No RIS algorithms, respectively, which means that $11.4\%, 18.8\%, 61.4\%$, and $62.5\%$ minimum SNRs obtained from the algorithms are lower than the target SNR, 20 dB. This validates the performance advantage of the proposed SUM and GSM algorithms. Moreover, the performance gains obtained from the proposed SUM and GSM algorithms are stable according to the CDF curve and consistent with the results in Fig. \ref{fig:Fig3}, which shows the superiority of the proposed algorithms.

Fig. \ref{fig:Fig4} plots users' minimum SNR across the number of reflective elements, i.e., $N$, where $P_B=10$ dBm and $P_S=0$ dBm. Here, we set $N_h=10$ and $N_v=4,7,10,13,16$. As expected, we observe that the minimum SNR increases as $N$ increases for the SUM and GSM algorithms. The performance gap between the proposed algorithms and benchmark schemes becomes larger as $N$ increases.

\subsection{RIS for Multiple-Antenna BS}
The proposed algorithms for multi-antenna BS, namely SUM-OB, SUM-MRB, GSM-OB, and GSM-MRB, are compared to the following benchmark schemes:
\begin{itemize}
\item Benchmark scheme 1 (No RIS): Without the RIS deployed and we only optimize the beamforming matrix at BS. The beamforming matrix is obtained by (\ref{A_structure}) and (\ref{A_ratio}).
\item Benchmark scheme 2 (RIS: random phase-MRB): The phase shifts are randomly chosen, whereas the beamforming matrix is obtained by (\ref{A_structure}) and (\ref{A_ratio}).
\end{itemize}

\begin{figure}[t]
 \centering
 \begin{minipage}[t]{0.5\textwidth}
  \centering
  \includegraphics[width=3in]{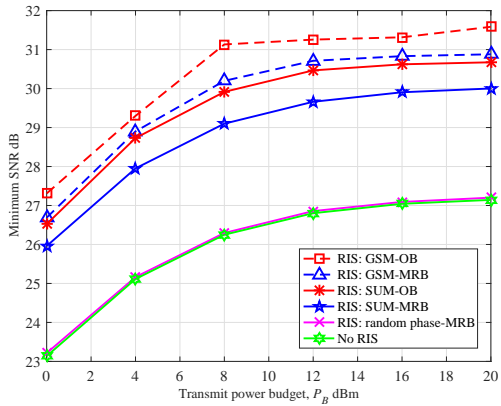}
  \caption{Minimum SNR over transmit power budget for various algorithms for the multiple-antenna BS case when $P_S=0$ dBm.}
  \label{fig:Fig5}
 \end{minipage}
 \begin{minipage}[t]{0.48\textwidth}
  \centering
  \includegraphics[width=3in]{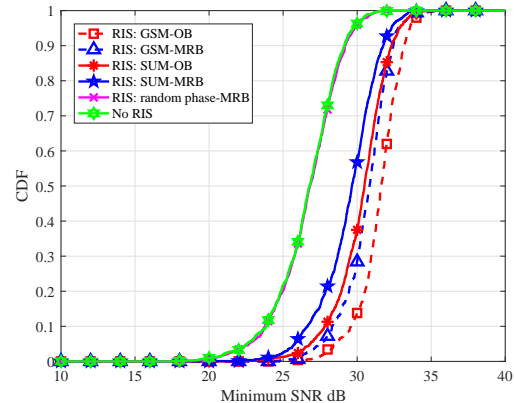}
  \caption{CDF of minimum SNR for different algorithms for the multiple-antenna BS case when $P_B=10$ dBm and $P_S=0$ dBm.}
  \label{fig:Fig6}
 \end{minipage}
\end{figure}

\begin{figure}[t]
 \centering
 \begin{minipage}[t]{0.5\textwidth}
  \centering
  \includegraphics[width=3in]{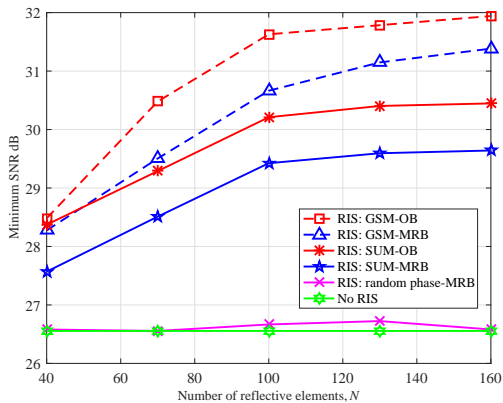}
  \caption{Minimum SNR over the number of the reflective elements when $P_B=10$ dBm and $P_S=0$ dBm.}
  \label{fig:Fig7}
 \end{minipage}
 \begin{minipage}[t]{0.48\textwidth}
  \centering
  \includegraphics[width=3in]{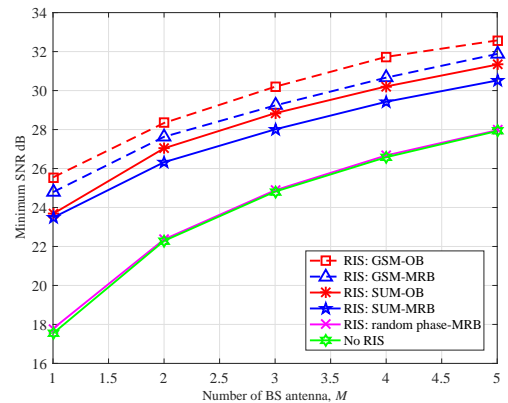}
  \caption{Minimum SNR over the number of the BS antennas, where $P_B=10$ dBm and $P_S=0$ dBm.}
  \label{fig:Fig8}
 \end{minipage}
\end{figure}

Fig. \ref{fig:Fig5} shows users' minimum SNR across the transmit power budget for various algorithms. It is observed that the minimum SNR increases up to a certain level as $P_B$ increases. As expected, the performances of the proposed algorithms outperform the benchmark schemes. Specifically, the SUM-MRB, SUM-OB, GSM-MRB, and GSM-OB algorithms can achieve the gain of approximately 2.9 dB, 3.5 dB, 3.7 dB, and 4.3 dB compared to the benchmark scheme without RIS. The GSM algorithms provide better performance compared to the SUM algorithms. The computational complexity reduction from MRB causes minimum SNR degradation marginally by approximately 0.7 dB, compared to the OB algorithm.

Fig. \ref{fig:Fig6} illustrates the CDF of the minimum SNR for different algorithms, similarly to the single-antenna BS case when $P_B=10$ dBm and $P_S=0$ dBm. It is observed that, for the same user success probability, two users can decode information with a much stringent SNR constraint with the proposed algorithms, compared to the benchmark schemes. The performance gains of the proposed algorithms are also stable according to the CDF curve, which means that the proposed algorithms can perform well with high probability.

Fig. \ref{fig:Fig7} shows users' minimum SNR over the number of reflective elements $N$ for the multiple-antenna BS case when $P_B=10$ dBm and $P_S=0$ dBm. The minimum SNRs for the proposed algorithms increase as $N$ increases, whereas the minimum SNRs of the benchmark schemes remain unchanged.

Fig. \ref{fig:Fig8} shows users' minimum SNR over the number of BS antennas, $M$, when $P_B=10$ dBm and $P_S=0$ dBm. As expected, the minimum SNR increases as $M$ increases for all schemes since the beamforming gain of BS increases. Here, it should be emphasized that the proposed algorithms outperform the benchmark schemes irrespective of $M$.

\section{Conclusions}\label{secconclusion}
 In this paper, an RIS-assisted TWRN was investigated and a joint beamforming and RIS design problem was formulated to maximize the minimum SNR under the transmit power constraint at the BS. The single-antenna BS case was first considered and addressed by devising the SUM and GSM algorithms. The optimization problem was then divided into two subproblems to design the phase shift and beamforming matrices for the case with a multiple-antenna BS. The RIS phase shift matrix was obtained by employing SUM or GSM method while the BS beamforming matrix was obtained by using OB or MRB method. Simulation results demonstrate that the proposed algorithms can achieve significant performance gains compared to the benchmark schemes, which validates the benefits of the RIS in TWRN.

\useRomanappendicesfalse
\appendices
\section{Proof of Proposition \ref{Proposition:Prop1}}\label{app:prop1}

\begin{proof}
From the Holder's inequality, (14) is bounded as follows:
\begin{align}
\gamma_1\leq  \frac{\beta\left\|(\boldsymbol{h}_1+\boldsymbol{V\Phi}_2\boldsymbol{g}_1)^T\boldsymbol{A}\right\|^2\left\|\boldsymbol{h}_2+\boldsymbol{V\Phi}_1\boldsymbol{g}_2\right\|^2}{\left\|(\boldsymbol{h}_1+\boldsymbol{V\Phi}_2\boldsymbol{g}_1)^T\boldsymbol{A}\right\|^2+1}=\frac{\beta\left\|\boldsymbol{h}_2+\boldsymbol{V\Phi}_1\boldsymbol{g}_2\right\|^2}{1+\frac{1}{\left\|(\boldsymbol{h}_1+\boldsymbol{V\Phi}_2\boldsymbol{g}_1)^T\boldsymbol{A}\right\|^2}}\tag{A1}
\end{align}
Here, the equality holds when $\boldsymbol{A}^H\hat{\boldsymbol{h}}_1^*$ is parallel to $\hat{\boldsymbol{h}}_2$, i.e.,
\begin{align}
    \boldsymbol{A}^H\hat{\boldsymbol{h}}_1^*=\mu \hat{\boldsymbol{h}}_2,\tag{A2}
\end{align}
where $\mu$ is a scalar, $\hat{\boldsymbol{h}}_1\triangleq\boldsymbol{h}_1+\boldsymbol{V\Phi}_2\boldsymbol{g}_1$, and $\hat{\boldsymbol{h}}_2\triangleq\boldsymbol{h}_2+\boldsymbol{V\Phi}_2\boldsymbol{g}_2$. Substituting the beamforming matrix structure $\boldsymbol{A}$ in \cite{liang2008optimal} into the equation, we obtain
\begin{align}
    (\hat{\boldsymbol{h}}_2\hat{\boldsymbol{h}}_1^T+\hat{\boldsymbol{h}}_1\hat{\boldsymbol{h}}_2^T)\hat{\boldsymbol{h}}_1^*= \mu \hat{\boldsymbol{h}}_2.\tag{A3}
\end{align}
The first term of the left-hand side in (A3) is scaled $\hat{\boldsymbol{h}}_2$ and the second term is  scaled $\hat{\boldsymbol{h}}_1$. The equality holds when $\hat{\boldsymbol{h}}_1 \perp \hat{\boldsymbol{h}}_2$ or $\hat{\boldsymbol{h}}_1 \parallel \hat{\boldsymbol{h}}_2$. $\hat{\boldsymbol{h}}_1$ and $\hat{\boldsymbol{h}}_2$ can be orthogonal when the number of the reflective elements is large. Similarly, we have
\begin{align}
\gamma_2 \leq \frac{\beta\left\|\boldsymbol{h}_1+\boldsymbol{V\Phi}_1\boldsymbol{g}_1\right\|^2}{1+\frac{1}{\left\|(\boldsymbol{h}_2+\boldsymbol{V\Phi}_2\boldsymbol{g}_2)^T\boldsymbol{A}\right\|^2}}.\tag{A4}
\end{align}
It is observed that $\boldsymbol{\Phi}_1$ and $\boldsymbol{\Phi}_2$ are decoupled in the SNR term and that their optimization can be split into two subproblems. For $\boldsymbol{\Phi}_1$, the larger $\left\|\boldsymbol{h}_2+\boldsymbol{V\Phi}_1\boldsymbol{g}_2\right\|^2$ makes $\gamma_1$ larger. Therefore, the optimal $\boldsymbol{\Phi}_1$ is the solution of the problem which
maximizes the minimum of $\left\|\boldsymbol{h}_1+\boldsymbol{V\Phi}_1\boldsymbol{g}_1\right\|^2$ and $\left\|\boldsymbol{h}_2+\boldsymbol{V\Phi}_1\boldsymbol{g}_2\right\|^2$. We have
\begin{align}
\left\|(\boldsymbol{h}_1+\boldsymbol{V\Phi}_2\boldsymbol{g}_1)^T\boldsymbol{A}\right\|^2=&\left\|\hat{\boldsymbol{h}}_1^T\boldsymbol{A}\right\|^2=\left\|\alpha\hat{\boldsymbol{h}}_1^T(\hat{\boldsymbol{h}}_2^*\hat{\boldsymbol{h}}_1^H+\hat{\boldsymbol{h}}_1^*\hat{\boldsymbol{h}}_2^H)\right\|^2\nonumber\\
=&\alpha^2\left\|\hat{\boldsymbol{h}}_1^T\hat{\boldsymbol{h}}_2^*\hat{\boldsymbol{h}}_1^H+\hat{\boldsymbol{h}}_1^T\hat{\boldsymbol{h}}_1^*\hat{\boldsymbol{h}}_2^H\right\|^2 \geq \alpha^2 |\hat{\boldsymbol{h}}_1^T\hat{\boldsymbol{h}}_1^*|^2\left\|\hat{\boldsymbol{h}}_2\right\|^2\nonumber\\
=&\alpha^2 |\hat{\boldsymbol{h}}_1^T\hat{\boldsymbol{h}}_1^*|^2\left\|\boldsymbol{h}_2+\boldsymbol{V\Phi}_2\boldsymbol{g}_2\right\|^2.\tag{A5}
\end{align}
Therefore, for $\boldsymbol{\Phi}_2$, the larger $\left\|\boldsymbol{h}_2+\boldsymbol{V\Phi}_2\boldsymbol{g}_2\right\|^2$ makes $\gamma_1$ larger. The optimal $\boldsymbol{\Phi}_2$ is also the solution of the problem which maximizes the minimum of $\left\|\boldsymbol{h}_1+\boldsymbol{V\Phi}_2\boldsymbol{g}_1\right\|^2$ and $\left\|\boldsymbol{h}_2+\boldsymbol{V\Phi}_2\boldsymbol{g}_2\right\|^2$. Therefore, the optimal RIS phase shift matrices in the first and second phases are identical to each other, i.e., $\boldsymbol{\Phi}_1=\boldsymbol{\Phi}_2$.
\end{proof}

\bibliography{reference}

\begin{thebibliography}{10}
\providecommand{\url}[1]{#1}
\csname url@rmstyle\endcsname
\providecommand{\newblock}{\relax}
\providecommand{\bibinfo}[2]{#2}
\providecommand\BIBentrySTDinterwordspacing{\spaceskip=0pt\relax}
\providecommand\BIBentryALTinterwordstretchfactor{4}
\providecommand\BIBentryALTinterwordspacing{\spaceskip=\fontdimen2\font plus
\BIBentryALTinterwordstretchfactor\fontdimen3\font minus
  \fontdimen4\font\relax}
\providecommand\BIBforeignlanguage[2]{{%
\expandafter\ifx\csname l@#1\endcsname\relax
\typeout{** WARNING: IEEEtran.bst: No hyphenation pattern has been}%
\typeout{** loaded for the language `#1'. Using the pattern for}%
\typeout{** the default language instead.}%
\else
\language=\csname l@#1\endcsname
\fi
#2}}

\bibitem{wang2020joint}
J.~Wang, Y.-C. Liang, X.~Yuan, and X.~Wang, ``Joint beamforming and
  reconfigurable intelligent surface design for two-way relay networks,'' in
  \emph{Proc. IEEE Globecom}, Taipei, Taiwan, China, Dec. 2020, pp. 1--6.

\bibitem{you2021towards}
X.~You \emph{et~al.}, ``Towards {6G} wireless communication networks: Vision,
  enabling technologies, and new paradigm shifts,'' \emph{Sci. China Inf.
  Sci.}, vol.~64, no.~1, pp. 1--74, 2021.

\bibitem{liang2020symbiotic}
Y.-C. Liang, Q.~Zhang, E.~G. Larsson, and G.~Y. Li, ``Symbiotic radio:
  Cognitive backscattering communications for future wireless networks,''
  \emph{IEEE Trans. Cogn. Commun. Netw.}, vol.~6, no.~4, pp. 1242--1255, 2020.

\bibitem{zhang20196g}
L.~Zhang, Y.-C. Liang, and D.~Niyato, ``{6G} visions: Mobile ultra-broadband,
  super internet-of-things, and artificial intelligence,'' \emph{China
  Commun.}, vol.~16, no.~8, pp. 1--14, Aug. 2019.

\bibitem{liang2020dynamic}
Y.-C. Liang, \emph{Dynamic Spectrum Management}.\hskip 1em plus 0.5em minus
  0.4em\relax Springer Nature, 2020.

\bibitem{liang2019large}
Y.-C. Liang, R.~Long, Q.~Zhang, J.~Chen, H.~V. Cheng, and H.~Guo, ``Large
  intelligent surface/antennas ({LISA}): Making reflective radios smart,''
  \emph{J. Commn. Inf. Netw.}, vol.~4, no.~2, pp. 40--50, June 2019.

\bibitem{gong2019towards}
S.~Gong, X.~Lu, D.~T. Hoang, D.~Niyato, L.~Shu, D.~I. Kim, and Y.-C. Liang,
  ``Toward smart wireless communications via intelligent reflecting surfaces: A
  contemporary survey,'' \emph{IEEE Commun. Surv. Tutor.}, vol.~22, no.~4, pp.
  2283--2314, 2020.

\bibitem{8796365}
E.~{Basar}, M.~{Di Renzo}, J.~{De Rosny}, M.~{Debbah}, M.~{Alouini}, and
  R.~{Zhang}, ``Wireless communications through reconfigurable intelligent
  surfaces,'' \emph{IEEE Access}, vol.~7, pp. 116\,753--116\,773, 2019.

\bibitem{huang2019reconfigurable}
C.~Huang, A.~Zappone, G.~C. Alexandropoulos, M.~Debbah, and C.~Yuen,
  ``Reconfigurable intelligent surfaces for energy efficiency in wireless
  communication,'' \emph{IEEE Trans. Wireless Commun.}, vol.~18, no.~8, pp.
  4157--4170, June 2019.

\bibitem{chen2019intelligent}
J.~Chen, Y.-C. Liang, Y.~Pei, and H.~Guo, ``Intelligent reflecting surface: A
  programmable wireless environment for physical layer security,'' \emph{IEEE
  Access}, vol.~7, pp. 82\,599--82\,612, 2019.

\bibitem{guan2019intelligent}
X.~Guan, Q.~Wu, and R.~Zhang, ``Intelligent reflecting surface assisted secrecy
  communication: Is artificial noise helpful or not?'' \emph{IEEE Wireless
  Commun. Lett.}, vol.~9, no.~6, pp. 778--782, 2020.

\bibitem{guo2020weighted}
H.~Guo, Y.-C. Liang, J.~Chen, and E.~G. Larsson, ``Weighted sum-rate
  maximization for reconfigurable intelligent surface aided wireless
  networks,'' \emph{IEEE Trans. Wireless Commun.}, Feb. 2020.

\bibitem{xie2020max}
H.~Xie, J.~Xu, and Y.-F. Liu, ``Max-min fairness in {IRS}-aided multi-cell
  {MISO} systems with joint transmit and reflective beamforming,'' \emph{IEEE
  Trans. Wireless Commun.}, vol.~20, no.~2, pp. 1379--1393, 2021.

\bibitem{xu2020resource}
D.~Xu, X.~Yu, Y.~Sun, D.~W.~K. Ng, and R.~Schober, ``Resource allocation for
  {IRS}-assisted full-duplex cognitive radio systems,'' \emph{IEEE Trans.
  Commun.}, vol.~68, no.~12, pp. 7376--7394, 2020.

\bibitem{han2019intelligent}
H.~Han, J.~Zhao, D.~Niyato, M.~Di~Renzo, and Q.-V. Pham, ``Intelligent
  reflecting surface aided network: Power control for physical-layer
  broadcasting,'' \emph{arXiv preprint arXiv:1910.14383}, 2019.

\bibitem{zhou2019intelligent}
G.~Zhou, C.~Pan, H.~Ren, K.~Wang, and A.~Nallanathan, ``Intelligent reflecting
  surface aided multigroup multicast {MISO} communication systems,'' \emph{IEEE
  Trans. Signal Process.}, vol.~68, pp. 3236--3251, 2020.

\bibitem{zhang2020large}
Q.~Zhang, Y.-C. Liang, and H.~V. Poor, ``Large intelligent surface/antennas
  ({LISA}) assisted symbiotic radio for {IoT} communications,'' \emph{arXiv
  preprint arXiv:2002.00340}, 2020.

\bibitem{yang2019intelligent}
G.~Yang, X.~Xu, and Y.-C. Liang, ``Intelligent reflecting surface assisted
  non-orthogonal multiple access,'' \emph{arXiv preprint arXiv:1907.03133},
  2019.

\bibitem{ding2019simple}
Z.~Ding and H.~V. Poor, ``A simple design of {IRS-NOMA} transmission,''
  \emph{IEEE Commun. Lett.}, vol.~24, no.~5, pp. 1119--1123, 2020.

\bibitem{wang2019intelligent}
P.~Wang, J.~Fang, X.~Yuan, Z.~Chen, and H.~Li, ``Intelligent reflecting
  surface-assisted millimeter wave communications: Joint active and passive
  precoding design,'' \emph{IEEE Trans. Veh. Technol.}, vol.~69, no.~12, pp.
  14\,960--14\,973, 2020.

\bibitem{perovic2019channel}
N.~S. Perovi{\'c}, M.~Di~Renzo, and M.~F. Flanagan, ``Channel capacity
  optimization using reconfigurable intelligent surfaces in indoor mmwave
  environments,'' \emph{arXiv preprint arXiv:1910.14310}, 2019.

\bibitem{rankov2007spectral}
B.~Rankov and A.~Wittneben, ``Spectral efficient protocols for half-duplex
  fading relay channels,'' \emph{IEEE J. Sel. Areas Commun.}, vol.~25, no.~2,
  pp. 379--389, Feb. 2007.

\bibitem{liang2008optimal}
Y.-C. Liang and R.~Zhang, ``Optimal analogue relaying with multi-antennas for
  physical layer network coding,'' in \emph{Proc. IEEE ICC}, Beijing, China,
  May 2008, pp. 3893--3897.

\bibitem{zhang2009optimal}
R.~Zhang, Y.-C. Liang, C.~C. Chai, and S.~Cui, ``Optimal beamforming for
  two-way multi-antenna relay channel with analogue network coding,''
  \emph{IEEE J. Sel. Areas Commun.}, vol.~27, no.~5, pp. 699--712, June 2009.

\bibitem{wang2012maximin}
W.~Wang, S.~Jin, and F.-C. Zheng, ``Maximin {SNR} beamforming strategies for
  two-way relay channels,'' \emph{IEEE Commun. Lett.}, vol.~16, no.~7, pp.
  1006--1009, May 2012.

\bibitem{fang2013beamforming}
Z.~Fang, X.~Wang, and X.~Yuan, ``Beamforming design for multiuser two-way
  relaying: A unified approach via max-min {SINR},'' \emph{IEEE Trans. Signal
  Process.}, vol.~61, no.~23, pp. 5841--5852, Sept. 2013.

\bibitem{joung2009multiuser}
J.~Joung and A.~H. Sayed, ``Multiuser two-way amplify-and-forward relay
  processing and power control methods for beamforming systems,'' \emph{IEEE
  Trans. Signal Process.}, vol.~58, no.~3, pp. 1833--1846, Dec. 2009.

\bibitem{zhang2020sum}
Y.~Zhang, C.~Zhong, Z.~Zhang, and W.~Lu, ``Sum rate optimization for two way
  communications with intelligent reflecting surface,'' \emph{IEEE Commun.
  Lett.}, vol.~24, no.~5, pp. 1090--1094, Mar. 2020.

\bibitem{atapattu2020reconfigurable}
S.~Atapattu, R.~Fan, P.~Dharmawansa, G.~Wang, J.~Evans, and T.~A. Tsiftsis,
  ``Reconfigurable intelligent surface assisted two--way communications:
  Performance analysis and optimization,'' \emph{IEEE Trans. Commun.}, vol.~68,
  no.~10, pp. 6552--6567, 2020.

\bibitem{peng2020multiuser}
Z.~Peng, Z.~Zhang, C.~Pan, L.~Li, and A.~L. Swindlehurst, ``Multiuser
  full-duplex two-way communications via intelligent reflecting surface,''
  \emph{IEEE Trans. Signal Process.}, vol.~69, pp. 837--851, 2021.

\bibitem{pradhan2020reconfigurable}
C.~Pradhan, A.~Li, L.~Song, B.~Vucetic, and Y.~Li, ``Hybrid precoding design
  for reconfigurable intelligent surface aided mmwave communication systems,''
  \emph{IEEE Wireless Commun. Lett.}, vol.~9, no.~7, pp. 1041--1045, 2020.

\bibitem{liu2019super}
H.~Liu, X.~Yuan, and Y.~J. Zhang, ``Super-resolution blind channel-and-signal
  estimation for massive {MIMO} with one-dimensional antenna array,''
  \emph{IEEE Trans. Signal Process.}, vol.~67, no.~17, pp. 4433--4448, July
  2019.

\bibitem{levin2012amplify}
G.~Levin and S.~Loyka, ``Amplify-and-forward versus decode-and-forward
  relaying: Which is better?'' in \emph{22th International Zurich seminar on
  communications (IZS)}.\hskip 1em plus 0.5em minus 0.4em\relax
  Eidgen{\"o}ssische Technische Hochschule Z{\"u}rich, 2012.

\bibitem{gao2009optimal}
F.~Gao, R.~Zhang, and Y.-C. Liang, ``Optimal channel estimation and training
  design for two-way relay networks,'' \emph{IEEE Trans. Commun.}, vol.~57,
  no.~10, pp. 3024--3033, 2009.

\bibitem{wang2016multipair}
S.~Wang, M.~Xia, and Y.-C. Wu, ``Multipair two-way relay network with
  harvest-then-transmit users: resolving pairwise uplink-downlink coupling,''
  \emph{IEEE J. Sel. Topics Signal Process.}, vol.~10, no.~8, pp. 1506--1521,
  2016.

\bibitem{nadeem2019intelligent}
Q.-U.-A. Nadeem, A.~Kammoun, A.~Chaaban, M.~Debbah, and M.-S. Alouini,
  ``Intelligent reflecting surface assisted multi-user {MISO} communication,''
  \emph{arXiv preprint arXiv:1906.02360}, 2019.

\bibitem{he2019cascaded}
Z.~{He} and X.~{Yuan}, ``Cascaded channel estimation for large intelligent
  metasurface assisted massive {MIMO},'' \emph{IEEE Wireless Commun. Lett.},
  vol.~9, no.~2, pp. 210--214, Oct. 2020.

\bibitem{taha2019enabling}
A.~Taha, M.~Alrabeiah, and A.~Alkhateeb, ``Enabling large intelligent surfaces
  with compressive sensing and deep learning,'' \emph{arXiv preprint
  arXiv:1904.10136}, 2019.

\bibitem{wang2020robust}
J.~Wang, Y.-C. Liang, S.~Han, and Y.~Pei, ``Robust beamforming and phase shift
  design for {IRS}-enhanced multi-user {MISO} downlink communication,'' in
  \emph{Proc. IEEE ICC}, Dublin, Ireland, June 2020, pp. 1--6.

\bibitem{grant2014cvx}
M.~Grant and S.~Boyd, ``{CVX}: {MATLAB} software for disciplined convex
  programming, version 2.1,'' 2014.

\bibitem{helmberg1996interior}
C.~Helmberg, F.~Rendl, R.~J. Vanderbei, and H.~Wolkowicz, ``An interior-point
  method for semidefinite programming,'' \emph{SIAM J. Optim}, vol.~6, no.~2,
  pp. 342--361, 1996.

\bibitem{joung2014energy}
J.~Joung, Y.~K. Chia, and S.~Sun, ``Energy-efficient, large-scale
  distributed-antenna system ({L-DAS}) for multiple users,'' \emph{J. Sel.
  Topics Signal Process.}, vol.~8, no.~5, pp. 954--965, Mar. 2014.

\bibitem{kammoun2020asymptotic}
A.~Kammoun, A.~Chaaban, M.~Debbah, M.-S. Alouini, \emph{et~al.}, ``Asymptotic
  max-min {SINR} analysis of reconfigurable intelligent surface assisted {MISO}
  systems,'' \emph{IEEE Trans. Wireless Commun.}, vol.~19, no.~12, pp.
  7748--7764, 2020.

\bibitem{access2010further}
``Further advancements for {E-UTRA} physical layer aspects,'' \emph{3GPP
  Technical Specification TR}, vol.~36, p.~V2, 2010.

\end{thebibliography}
\bibliographystyle{IEEEtran}
\end{document}